\documentclass[11pt]{article}

\usepackage[english]{babel}

\usepackage[letterpaper,top=2cm,bottom=1.5cm,left=1in,right=1in,marginparwidth=1.75cm]{geometry}


\usepackage{setspace}
\usepackage{amsmath}
\usepackage{amsthm,thm-restate}
\usepackage{graphicx}
\usepackage{amssymb}
\usepackage{wrapfig}
\usepackage[compact]{titlesec}
\usepackage[colorlinks=true, allcolors=blue]{hyperref}
\usepackage[capitalise,nameinlink]{cleveref}
\usepackage[textsize=tiny]{todonotes}

\Crefname{algocf}{Algorithm}{Algorithms}
\crefname{algocfline}{line}{lines}
\Crefname{invariant}{Invariant}{Invariants}
\Crefname{claim}{Claim}{Claims}
\Crefname{subclaim}{Subclaim}{Subclaims}

\usepackage[ruled,vlined,linesnumbered,algonl]{algorithm2e}
\SetEndCharOfAlgoLine{}
\SetKwComment{Comment}{\footnotesize$\triangleright$\ }{}

\SetCommentSty{mycommfont}

\usepackage{tikz}

\usepackage{nicefrac}
\newcommand{\nf}{\nicefrac}
\usepackage[shortlabels]{enumitem}
\usepackage[font={small,it}]{caption}

\newtheorem{theorem}{Theorem}[section]
\newtheorem{lemma}[theorem]{Lemma}
\newtheorem{definition}[theorem]{Definition}
\newtheorem{corollary}[theorem]{Corollary}

\newtheorem{claim}[theorem]{Claim}

\newcommand{\sse}{\subseteq}

\newcommand{\ZZ}{\mathbb{Z}}

\newcommand{\cE}{\mathcal{E}}
\newcommand{\subtree}{\Gamma}
\newcommand{\vote}{\gamma}

\newcommand{\eT}{T}
\newcommand{\xT}{\overline{T}}
\newcommand{\ST}{C}

\newcommand{\e}{\varepsilon}

\newcommand{\ExEx}{\operatorname{ExtraExp}}
\newcommand{\cost}{M}
\newcommand{\mem}{\operatorname{mem}}

\newcommand{\TreeX}{\textsc{TreeX}\xspace}
\newcommand{\AlgoKnownD}{\textsc{TreeX-KnownDist}\xspace}
\newcommand{\FullInfoX}{\textsc{FullInfoX}\xspace}

\title{Graph Searching with Predictions}

\author{
  Siddhartha Banerjee \and
  Vincent Cohen-Addad \and
  Anupam Gupta \and
  Zhouzi Li
}

\begin{document}
\maketitle

\begin{abstract}
  Consider an agent %(say a robot) 
  exploring an unknown graph in search
  of some goal state. As it walks around the graph, it learns the
  nodes and their neighbors. The agent only knows where the goal state
  is when it reaches it. How do we reach this goal while moving only a
  small distance? This problem seems hopeless, even on trees of
  bounded degree, unless we give the agent some help. This setting
  with ``help'' often arises in exploring large search spaces (e.g.,
  huge game trees) where we assume access to some score/quality
  function for each node, which we use to guide us towards the
  goal. In our case, we assume the help comes in the form of \emph{distance
    predictions}: each node $v$ provides a prediction $f(v)$ of its
  distance to the goal vertex.  Naturally if these predictions are
  correct, we can reach the goal along a shortest path. What if the
  predictions are unreliable and some of them are erroneous? Can we
  get an algorithm whose performance relates to the error of the
  predictions?

  \medskip In this work, we consider the problem on trees  and give deterministic
  algorithms whose total movement cost is only $O(OPT + \Delta \cdot ERR)$,
  where $OPT$ is the distance from the start to the goal vertex,
  $\Delta$ the maximum degree, and the $ERR$ is the total number of
  vertices whose predictions are erroneous. We show this guarantee is
  optimal. We then consider a ``planning'' version of the problem
  where the graph and predictions are known at the beginning, so the
  agent can use this global information to devise a search strategy of
  low cost. For this planning version, we go beyond trees and give an
  algorithms which gets good performance on (weighted) graphs with bounded doubling dimension.
\end{abstract}

%\balert{``explored'' is now ``visited''. ``revealed'' is now ``observed''.}
%\alert{Using both $X$ and $\cE$ for the erroneous nodes. Change to
%  $\cE$. Also remove the notation $\error$ and $\error^t$.}
%\alert{Replace $k$ by $\Delta$, and let $\Delta_v$ be the degree of $v$.}
%\alert{Integer edge lengths and integer predictions, right?}
%\alert{Use $\ell(\cdot)$ instead of $l(\cdot)$ for level}

\section{Introduction}
\label{sec:introduction}

Consider an agent (say a robot) traversing an environment modeled as
an undirected graph $G = (V,E)$. It starts off at some \emph{root}
vertex $r$, and commences looking for a goal vertex $g$. However, the
location of this goal is initially unknown to the agent, who gets to
know it only when it visits vertex $g$. So the agent starts exploring
from $r$, 
visits various vertices $r = v_0, v_1, \cdots, v_t, \cdots$ in $G$ one
by one, until it reaches $g$. The cost it incurs at timestep $t$ is
the distance it travels to get from $v_{t-1}$ to $v_t$. How can the
agent minimize the total cost? This framework is very general,
capturing not only problems in robotic exploration, but also general
questions related to game tree search: how to reach a goal state with
the least effort? 

% \alert{More here, give
%   citations, connect to $A^*$ search if possible, etc.}

Since this is a question about optimization under uncertainty, we use
the notion of \emph{competitive analysis}: we relate the cost incurred
by the algorithm on an instance to the optimal cost incurred in
hindsight. The latter is just the distance $D := d(r,g)$ between the
start and goal vertices. Sadly, a little thought tells us that this
problem has very pessimistic guarantees in the absence of any further
constraints. For example, even if the graph is known to be a complete
binary tree and the goal is known to be at some distance $D$ from the
root, the adversary can force any algorithm to incur an expected cost
of $\Omega(2^D)$. Therefore the competitiveness is unbounded as $D$
gets large. This is why previous works in online algorithms enforced
topological constraints on the graph, such as restricting the graph to
be a path, or $k$ paths meeting at the root, or a
grid~\cite{Baezayates1993234}.

But in many cases (such as in game-tree search) we want to solve this
problem for broader classes of graphs---say for complete binary trees
(which were the bad example above), or even more general settings. The
redeeming feature in these settings is that we are not searching
blindly: the nodes of the graph come with estimates of their quality,
which we can use to search effectively. % Such suggestions are used in
% popular search algorithms like $A^*$-search. 
What are good algorithms in such settings? What can we prove about them?

In this paper we formalize these questions via the idea of
\emph{distance predictions}: each node $v$
gives a prediction $f(v)$ of its distance $d_G(v,g)$ to the goal
state. If these predictions are all correct, we can just ``walk
downhill''---i.e., starting with $v_0$ being the start node, we can
move at each timestep $t$ to a neighbor $v_t$ of $v_{t-1}$ with
$f(v_t) = f(v_{t-1})-1$. This reaches the goal along a shortest
path. However, getting perfect predictions seems unreasonable, so we
ask:
\begin{quote}
  \emph{What if a few of the predictions are incorrect?}  Can we
  achieve an ``input-sensitive'' or ``smooth'' or ``robust'' bound,
  where we incur a cost of $d(g,r) +$ some function of the
  prediction error?
\end{quote}

We consider two versions of the problem:
\begin{description}
\item \textbf{The Exploration Problem.} In this setting the graph $G$
  is initially unknown to the agent: it only knows the vertex
  $v_0 = r$, its neighbors $\partial v_0$, and the predictions on all
  these nodes. In general, at the beginning of time $t \geq 1$, it
  knows the vertices $V_{t-1} = \{v_0, v_1, \cdots, v_{t-1}\}$ visited
  in the past, all their neighboring vertices $\partial V_{t-1}$, and
  the predictions for all the vertices in
  $V_{t-1} \cup \partial V_{t-1}$. The agent must use this information
  to move to some unvisited neighbor (which is now called $v_t$),
  paying a cost of $d(v_{t-1}, v_t)$. It then observes the edges
  incident to $v_t$, along with the predictions for nodes newly
  observed.
  
\item \textbf{The Planning Problem.} This is a simpler version of
  the problem where the agent starts off knowing the entire graph $G$,
  as well as the predictions at all its nodes. It just does not know
  which node is the goal, and hence it must traverse the graph in some
  order. 
\end{description}
The cost in both cases is the total distance traveled by the agent
until it reaches the goal, and the competitive ratio is the ratio of this
quantity to the shortest path distance $d(r,g)$ from the root to the
goal. 

\subsection{Our Results}

Our first main result is for the (more challenging) exploration problem,
for the case of trees.

\begin{restatable}[Exploration]{theorem}{UnknownD}
  \label{thm:main}
  The (deterministic) \TreeX algorithm solves the graph exploration
  problem on trees in the presence of predictions: on any (unweighted)
  tree with maximum degree $\Delta$, for any constant $\delta>0$, the algorithm incurs a cost of
  \[ d(r,g)(1+\delta) + O(\Delta \cdot |\cE|/\delta), \] where % $D := d(r,g)$ is the
  % distance of the goal from the starting vertex, and
  $\cE := \{v \in V \mid f(v) \neq d(v,g)\}$ is the set of vertices
  that give erroneous predictions.
\end{restatable}

One application of the above theorem is for the layered graph traversal problem (see \S\ref{sec:related-work} for a complete definition). 

\begin{corollary}[Robustness and Consistency for the Layered Graph Traversal problem.]
There exists an algorithm that achieves the following guarantees for the layered graph traversal problem in the presence of predictions: given an instance with maximum degree $\Delta$ and width $k$, for any constant $\delta > 0$, the algorithm incurs an expected cost of at most
$\min(O(k^2 \log \Delta) \, OPT, OPT + O(\Delta |\cE|))$.
% \todo{VCA: Added this
% here, feel free to remove.}
\end{corollary}

The proof of the above corollary is immediate: Since the
input is a tree (with some additional structure that we do not
require) that is revealed online, we can use the algorithm from \Cref{thm:main}. Hence, given an
instance $I$ of layered graph traversal (with predictions), we can use
the algorithm from \Cref{thm:main} in combination with the
\cite{https://doi.org/10.48550/arxiv.2202.04551}, thereby being
both \emph{consistent} and \emph{robust}: if the predictions are of high quality, then our algorithm ensures that the cost will be nearly optimal; otherwise if the predictions are useless, \cite{https://doi.org/10.48550/arxiv.2202.04551}'s algorithm gives an upper bound in the worst case.

Moreover, we show that the guarantee
obtained in  \Cref{thm:main}  is the best possible, up to
constants. 

\begin{theorem}[Exploration Lower Bound]
  \label{thm:lower-bound}
  Any algorithm (even randomized) for the graph exploration problem
  with predictions must incur a cost of at least
  $\max(d(r,g),\Omega(\Delta\cdot |\cE|))$. 
\end{theorem}

\begin{proof}
  The lower bound of $d(r,g)$ is immediate. For the second term,
  consider the setting where the root $r$ has $\Delta$ disjoint paths
  of length $D$ leaving it, and the goal is guaranteed to be at one of
  the leaves. Suppose we are given the ``null'' prediction, where each vertex
  predicts $f(v) = D+\ell(v)$ (where $\ell(v)$ is the distance of the vertex from the
  root, which we henceforth refer to as the \emph{level} of the vertex). The erroneous vertices are the $D$ vertices along the $r$-$g$
  path. Since the predictions do not give any signal at all (they can
  be generated by the algorithm itself), this is a problem of guessing
  which of the leaves is the goal, and any algorithm, even randomized,
  must travel $\Omega(\Delta\cdot D) = \Omega(\Delta \cdot |\cE|)$
  before reaching the goal.
  % \alert{Is it possible that even for
  %   unknown-$D$, we can get $D + O(\Delta |\cE|)$, or do we need to
  %   lose some constant in front of the $D$?}
\end{proof}

% \todo{AG: Should we give the full-information algorithms as well? Maybe
%   so, given the doubling dimension algorithm.}

% \medskip\textbf{The Planning Problem.} 
Our next set of results are for the planning problem, where we know
the graph and the predictions up-front, and must come up with a
strategy with this global information.
% , we can give results
% for a broader set of graphs.

\begin{restatable}[Planning]{theorem}{FullInfo}
  \label{thm:full-info}
  % For tree exploration in the full-information case, there is an
  % algorithm that incurs cost at most $d(r,g) + O(\Delta \cdot |\cE|)$.
  For the planning version of the graph exploration problem, % in the full-information case, 
  there is an algorithm that incurs cost at most % \alert{Are there
    % constants here?}
  \begin{enumerate}[label=(\roman*), nolistsep]
  \item $d(r,g) + O(\Delta \cdot |\cE|)$ if the graph is a tree, where $\Delta$ is the maximal degree. 
  \item $d(r,g) + 2^{O(\alpha)} \cdot O(|\cE|^2)$ where $\alpha$ is the doubling dimension of $G$.
  \end{enumerate}
  Again, $\cE$ is the set of nodes with incorrect predictions.
\end{restatable}

Note that result~(i) is very similar to that of~\Cref{thm:main} (for
the harder exploration problem): the differences are that we do not
lose any constant in the distance $d(r,g)$ term, and also that the 
algorithm used here (for the planning problem) is simpler. Moreover,
the lower bound from \Cref{thm:lower-bound} continues to hold in the
planning setting, since the knowledge of the graph and the predictions
does not help the algorithm; hence result~(i) is tight.

We do not yet know an analog of result~(ii) for 
the exploration problem: extending \Cref{thm:main} to general graphs, even those with bounded doubling metrics remains a tantalizing open
problem. Moreover, we currently do not have a lower bound matching
result~(ii); indeed, we conjecture that a cost of 
$d(r,g) + 2^{O(\alpha)} \cdot |\cE|$ should be achievable.
We leave these as questions for future investigation.

\subsection{Our Techniques}
\label{sec:our-techniques}

To get some intuition for the problem, consider the case where given a
tree and a guarantee that the goal is at distance $D$ from the start
node $r$. Suppose each node $v$ gives the ``null'' prediction of
$f(v)=D+d(r,v)$. In case the tree is a complete binary tree, then
these predictions carry no information and we would have to
essentially explore all nodes within distance $D$. But note that the
predictions for about half of these nodes are incorrect, so these
erroneous nodes can pay for this exploration. But now consider a
``lopsided'' example, with a binary tree on one side of the root, and
a path on the other (\Cref{fig: lopsided tree}). Suppose the goal is
at distance $D$ along the path. In this case, only the path nodes are
incorrect, and we only have $O(D + |\cE|) = O(D)$ budget for the
exploration. In particular, we must explore more aggressively along
the path, and balance the exploration on both sides of the root. But
such gadgets can be anywhere in the tree, and the predictions can be
far more devious than the null-prediction, so we need to generalize
this idea.

% : the root $r$
% has two children $a,b$, and the subtree rooted on $a$ is a complete
% binary tree while the subtree rooted on $b$ is a path to the goal
% $g$). Now the ``null'' prediction is almost correct except for the
% path from $b$ to $g$, so the ``budget'' we have to explore the subtree
% rooted on $a$ is very limitted (which is $O(D)$). To solve this
% particular example, the idea is to balance the exploration and the
% budget by balancing the exploration on both sides of the root. We will
% show how this idea can be generalized.

\begin{wrapfigure}{R}{0.4\textwidth}
  \centering
  \scalebox{1}{    \tikzset{every picture/.style={line width=0.75pt}} %set default line width to 0.75pt        

\begin{tikzpicture}[x=0.75pt,y=0.75pt,yscale=-1,xscale=1]

%Straight Lines [id:da3502261274657965] 
\draw  [dash pattern={on 4.5pt off 4.5pt}]  (210,150) -- (257.27,156.6) ;
%Straight Lines [id:da5488093106752505] 
\draw    (183.7,106.3) -- (210,100) ;
%Straight Lines [id:da9928993554645467] 
\draw    (183.7,140) -- (210,150) ;
%Straight Lines [id:da7769146975272274] 
\draw    (183.7,140) -- (210,130) ;
%Straight Lines [id:da127089069428717] 
\draw  [dash pattern={on 4.5pt off 4.5pt}]  (210,100) -- (256.3,93.7) ;
%Straight Lines [id:da060249439946234995] 
\draw    (143.7,116.3) -- (183.7,106.3) ;
%Straight Lines [id:da756000937076285] 
\draw  [dash pattern={on 4.5pt off 4.5pt}]  (210,190) -- (257.6,209.3) ;
%Straight Lines [id:da18640743388188086] 
\draw  [dash pattern={on 4.5pt off 4.5pt}]  (210,190) -- (257.6,196.3) ;
%Straight Lines [id:da943665880637381] 
\draw  [dash pattern={on 4.5pt off 4.5pt}]  (210,170) -- (257.6,183.3) ;
%Straight Lines [id:da7430182513639698] 
\draw  [dash pattern={on 4.5pt off 4.5pt}]  (210,170) -- (257.6,169.3) ;
%Straight Lines [id:da8465128648518967] 
\draw    (183.7,172.6) -- (210,190) ;
%Straight Lines [id:da4343169972639058] 
\draw    (183.7,172.6) -- (210,170) ;
%Straight Lines [id:da5790388918884821] 
\draw    (143.7,153.7) -- (183.7,140) ;
%Straight Lines [id:da5449392363894787] 
\draw    (143.7,153.7) -- (183.7,172.6) ;
%Shape: Circle [id:dp5058227324758564] 
\draw  [fill={rgb, 255:red, 208; green, 2; blue, 27 }  ,fill opacity=1 ] (180,172.6) .. controls (180,170.56) and (181.66,168.9) .. (183.7,168.9) .. controls (185.74,168.9) and (187.4,170.56) .. (187.4,172.6) .. controls (187.4,174.64) and (185.74,176.3) .. (183.7,176.3) .. controls (181.66,176.3) and (180,174.64) .. (180,172.6) -- cycle ;
%Shape: Circle [id:dp885797131921485] 
\draw  [fill={rgb, 255:red, 208; green, 2; blue, 27 }  ,fill opacity=1 ] (180,140) .. controls (180,137.96) and (181.66,136.3) .. (183.7,136.3) .. controls (185.74,136.3) and (187.4,137.96) .. (187.4,140) .. controls (187.4,142.04) and (185.74,143.7) .. (183.7,143.7) .. controls (181.66,143.7) and (180,142.04) .. (180,140) -- cycle ;
%Straight Lines [id:da810969790512035] 
\draw    (95,135) -- (143.7,116.3) ;
%Straight Lines [id:da9744059835403769] 
\draw    (95,135) -- (143.7,153.7) ;
%Shape: Circle [id:dp9995714170125762] 
\draw  [fill={rgb, 255:red, 74; green, 144; blue, 226 }  ,fill opacity=1 ] (91.3,135) .. controls (91.3,132.96) and (92.96,131.3) .. (95,131.3) .. controls (97.04,131.3) and (98.7,132.96) .. (98.7,135) .. controls (98.7,137.04) and (97.04,138.7) .. (95,138.7) .. controls (92.96,138.7) and (91.3,137.04) .. (91.3,135) -- cycle ;
%Shape: Circle [id:dp31859699968268185] 
\draw  [fill={rgb, 255:red, 208; green, 2; blue, 27 }  ,fill opacity=1 ] (140,153.7) .. controls (140,151.66) and (141.66,150) .. (143.7,150) .. controls (145.74,150) and (147.4,151.66) .. (147.4,153.7) .. controls (147.4,155.74) and (145.74,157.4) .. (143.7,157.4) .. controls (141.66,157.4) and (140,155.74) .. (140,153.7) -- cycle ;
%Shape: Circle [id:dp6481306606579553] 
\draw  [fill={rgb, 255:red, 208; green, 2; blue, 27 }  ,fill opacity=1 ] (140,116.3) .. controls (140,114.26) and (141.66,112.6) .. (143.7,112.6) .. controls (145.74,112.6) and (147.4,114.26) .. (147.4,116.3) .. controls (147.4,118.34) and (145.74,120) .. (143.7,120) .. controls (141.66,120) and (140,118.34) .. (140,116.3) -- cycle ;
%Shape: Circle [id:dp16144546507645585] 
\draw  [fill={rgb, 255:red, 255; green, 255; blue, 255 }  ,fill opacity=1 ] (253.9,183.3) .. controls (253.9,181.26) and (255.56,179.6) .. (257.6,179.6) .. controls (259.64,179.6) and (261.3,181.26) .. (261.3,183.3) .. controls (261.3,185.34) and (259.64,187) .. (257.6,187) .. controls (255.56,187) and (253.9,185.34) .. (253.9,183.3) -- cycle ;
%Shape: Circle [id:dp23273245955229527] 
\draw  [fill={rgb, 255:red, 255; green, 255; blue, 255 }  ,fill opacity=1 ] (253.9,169.3) .. controls (253.9,167.26) and (255.56,165.6) .. (257.6,165.6) .. controls (259.64,165.6) and (261.3,167.26) .. (261.3,169.3) .. controls (261.3,171.34) and (259.64,173) .. (257.6,173) .. controls (255.56,173) and (253.9,171.34) .. (253.9,169.3) -- cycle ;
%Shape: Circle [id:dp039510356342447484] 
\draw  [fill={rgb, 255:red, 255; green, 255; blue, 255 }  ,fill opacity=1 ] (253.9,209.3) .. controls (253.9,207.26) and (255.56,205.6) .. (257.6,205.6) .. controls (259.64,205.6) and (261.3,207.26) .. (261.3,209.3) .. controls (261.3,211.34) and (259.64,213) .. (257.6,213) .. controls (255.56,213) and (253.9,211.34) .. (253.9,209.3) -- cycle ;
%Shape: Circle [id:dp23540474190362826] 
\draw  [fill={rgb, 255:red, 255; green, 255; blue, 255 }  ,fill opacity=1 ] (253.9,196.3) .. controls (253.9,194.26) and (255.56,192.6) .. (257.6,192.6) .. controls (259.64,192.6) and (261.3,194.26) .. (261.3,196.3) .. controls (261.3,198.34) and (259.64,200) .. (257.6,200) .. controls (255.56,200) and (253.9,198.34) .. (253.9,196.3) -- cycle ;
%Shape: Circle [id:dp717288341163713] 
\draw  [fill={rgb, 255:red, 208; green, 2; blue, 27 }  ,fill opacity=1 ] (180,106.3) .. controls (180,104.26) and (181.66,102.6) .. (183.7,102.6) .. controls (185.74,102.6) and (187.4,104.26) .. (187.4,106.3) .. controls (187.4,108.34) and (185.74,110) .. (183.7,110) .. controls (181.66,110) and (180,108.34) .. (180,106.3) -- cycle ;
%Shape: Circle [id:dp41810175010079087] 
\draw  [fill={rgb, 255:red, 126; green, 211; blue, 33 }  ,fill opacity=1 ] (252.6,93.7) .. controls (252.6,91.66) and (254.26,90) .. (256.3,90) .. controls (258.34,90) and (260,91.66) .. (260,93.7) .. controls (260,95.74) and (258.34,97.4) .. (256.3,97.4) .. controls (254.26,97.4) and (252.6,95.74) .. (252.6,93.7) -- cycle ;
%Straight Lines [id:da09568913375056454] 
\draw  [dash pattern={on 4.5pt off 4.5pt}]  (210,150) -- (257.27,142.6) ;
%Straight Lines [id:da4055063901907552] 
\draw  [dash pattern={on 4.5pt off 4.5pt}]  (210,130) -- (257.27,130.6) ;
%Straight Lines [id:da8753953723526702] 
\draw  [dash pattern={on 4.5pt off 4.5pt}]  (210,130) -- (257.27,116.6) ;
%Shape: Circle [id:dp3783820969346332] 
\draw  [fill={rgb, 255:red, 255; green, 255; blue, 255 }  ,fill opacity=1 ] (253.57,130.6) .. controls (253.57,128.56) and (255.22,126.9) .. (257.27,126.9) .. controls (259.31,126.9) and (260.97,128.56) .. (260.97,130.6) .. controls (260.97,132.64) and (259.31,134.3) .. (257.27,134.3) .. controls (255.22,134.3) and (253.57,132.64) .. (253.57,130.6) -- cycle ;
%Shape: Circle [id:dp0669935352611224] 
\draw  [fill={rgb, 255:red, 255; green, 255; blue, 255 }  ,fill opacity=1 ] (253.57,116.6) .. controls (253.57,114.56) and (255.22,112.9) .. (257.27,112.9) .. controls (259.31,112.9) and (260.97,114.56) .. (260.97,116.6) .. controls (260.97,118.64) and (259.31,120.3) .. (257.27,120.3) .. controls (255.22,120.3) and (253.57,118.64) .. (253.57,116.6) -- cycle ;
%Shape: Circle [id:dp18137912432802783] 
\draw  [fill={rgb, 255:red, 255; green, 255; blue, 255 }  ,fill opacity=1 ] (253.57,156.6) .. controls (253.57,154.56) and (255.22,152.9) .. (257.27,152.9) .. controls (259.31,152.9) and (260.97,154.56) .. (260.97,156.6) .. controls (260.97,158.64) and (259.31,160.3) .. (257.27,160.3) .. controls (255.22,160.3) and (253.57,158.64) .. (253.57,156.6) -- cycle ;
%Shape: Circle [id:dp20807886591516933] 
\draw  [fill={rgb, 255:red, 255; green, 255; blue, 255 }  ,fill opacity=1 ] (253.57,142.6) .. controls (253.57,140.56) and (255.22,138.9) .. (257.27,138.9) .. controls (259.31,138.9) and (260.97,140.56) .. (260.97,142.6) .. controls (260.97,144.64) and (259.31,146.3) .. (257.27,146.3) .. controls (255.22,146.3) and (253.57,144.64) .. (253.57,142.6) -- cycle ;

% Text Node
\draw (80,130) node [anchor=north west][inner sep=0.75pt]  [font=\small]  {$r$};
% Text Node
\draw (130,155) node [anchor=north west][inner sep=0.75pt]  [font=\small]  {$a$};
% Text Node
\draw (130,105) node [anchor=north west][inner sep=0.75pt]  [font=\small]  {$b$};
% Text Node
\draw (265,90) node [anchor=north west][inner sep=0.75pt]  [font=\small]  {$g$};

\end{tikzpicture}}
  \caption{The subtree rooted on $r$'s child $a$ is a complete binary tree, while the subtree rooted on $b$ is a path to the goal $g$. ``Null'' predictions $f(v)=D+d(r,v)$ at every $v$ have a total error $D$ (only nodes on the path from $r$ to $g$ have errors on predictions).}
\vspace{-0.5cm}
\label{fig: lopsided tree}
\end{wrapfigure}
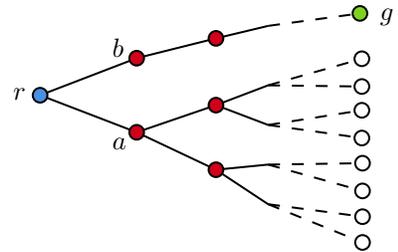

We start off with a special case which we call the
\emph{known-distance} case.  This is almost the same as the general
problem, but with the additional guarantee that the prediction of the
root is correct. Equivalently, we are given the distance $D := d(r,g)$
of the goal vertex from the root/starting node $r$. For this setting,
we can get the following very sharp result:

\begin{restatable}[Known-Distance Case]{theorem}{KnownD}
  \label{thm:known-D}
  The \AlgoKnownD algorithm solves the graph exploration problem
  in the known-distance case, incurring a cost of at most $d(r,g) +
  O(\Delta |\cE|)$.
\end{restatable}

Hence in the zero-error case, or in low-error cases where
$|\cE| \ll D$, the algorithm loses very little compared to the
optimal-in-hindsight strategy, which just walks from the root to the
goal vertex, and incurs a cost of $D$. This algorithm is inspired by
the ``lopsided'' example above: it not only balances the exploration on
different subtrees, but also at multiple levels. To generalize this
idea from predictions, we introduce the concepts of \emph{anchor} and
\emph{loads} (see \S\ref{sec:definitions}). At a high level, for each
node we consider the subtrees rooted at its children, and identify
subset of nodes in each of these subtrees which are erroneous
depending on which subtree contains the goal $g$. We ensure that these
sets have near-equal sizes, so that no matter which of these subtrees
contains the goal, one of them can pay for the others. This requires
some delicacy, since we need to ensure this property throughout the
tree. The details appear in \S\ref{sec:known-D}.

Having proved \Cref{thm:known-D}, we use the algorithm to then solve
the problem where the prediction for the root vertex may itself be
erroneous. Given \Cref{thm:known-D} and \Cref{alg:degree-k}, we can
reduced the problem to finding some node $v$ such that $d(v,g)$ is
known; moreover this $v$ must not be very far from the start node
$r$. The idea is conceptually simple: as we explore the graph, if most
predictions are correct we can use these predictions to find such a
$v$, otherwise these incorrect predictions give us more budget to
continue exploring. Implementing this idea (and particularly, doing
this deterministically) requires us to figure out how to ``triangulate''
with errors, which we do in \S\ref{sec: alg for unknown D}.

Finally, we give the ideas behind the algorithms for the
\emph{planning version} of the problem. The main idea is to define the
implied-error function $\varphi(v):=|\{u\mid f(u)\neq d(u,v)\}|$,
which measures the error if the goal is sitting at node $v$. Since we
know all the predictions and the tree structure in this version of the
problem, and moreover $\phi(g)=|\cE|$, it is natural to search the
graph greedily in increasing order of the implied-error. However,
naively doing this may induce a large movement cost, so we bucket
nodes with similar implied-error together, and then show that the
total cost incurred in exploring all nodes with
$\varphi(v) \approx 2^i$ is itself close to $2^i$ (times a factor that
depends on the degree or the doubling dimension). It remains an
interesting open problem to extend this algorithm to broader classes
of graphs. The details here appear in \S\ref{sec:full-information}.  % We prove the upper
% bound for the trees and graphs with bounded doubling dimension, but we
% conjecture that the algorithm works for general graphs.  and
% \alert{blah blah}.

\subsection{Related Work}
\label{sec:related-work}

\textbf{Graph Searching.} Graph searching is a fundamental problem, and there are too many
variants to comprehensively discuss: we point to the works closest to
ours.
% The paper of 
% \href{https://sites.cs.ucsb.edu/~suri/cs231/Rlist/search.pdf}{
Baeza-Yates, Culberson, and Rawlins~\cite{Baezayates1993234}
considered the exploration problem without predictions on the line
(where it is also called the ``cow-path'' problem),
on $k$-spiders (i.e., where $k$ semi-infinite lines meet at the root) and in the plane:
% (without any predictions): 
they showed tight bounds of $9$ on the
deterministic competitive ratio of the line, $1 + 2k^k/(k-1)^{k-1}$
for $k$-spiders, among other results. Their lower bounds (given for
``monotone-increasing strategies'') were generalized by Jaillet and
Stafford~\cite{JailletS01}; \cite{JailletSG02} point out that the
results for $k$-spiders were obtained by Gal~\cite{MR585693} before~\cite{Baezayates1993234} (see
also~\cite{0016627}). Kao et al.~\cite{KaoRT96,KaoMSY98} give tight
bounds for both deterministic and randomized algorithms, even with multiple agents.

The \emph{layered graph traversal} problem~\cite{PAPADIMITRIOU1991127}
is very closely related to our model. A tree is revealed over time. At
each timestep, some of the leaves of the current tree \emph{die}, and
others have some number of children. The agent is required to sit at
one of the current (living) leaves, so if the node the agent previously sat is no longer a leaf or is dead, the agent is forced to move.
% the graph is guaranteed to be 
% layered graph, where the initial layer $L^0$ just contains the root,
% and edges go between nodes in adjacent layers. All the nodes $L^{i+1}$
% of layer $i+1$ (and the edges between them and nodes $L^i$ at layer
% $i$) are revealed when the searcher reaches a node in $L^i$ for the
% first time. 
The game ends when the goal state is revealed and objective is to minimize the total movement cost. The \emph{width} $k$ of
the problem is the largest number of leaves alive at any time (observe that
we do not parameterize our algorithm with this parameter). This
problem is essentially the cow-path problem for the case of $w=2$, but is
substantially more difficult for larger widths. Indeed, the
deterministic bounds lie between $\Omega(2^k)$~\cite{FiatFKRRV98} and
$O(k 2^k)$~\cite{Burley96}. Ramesh~\cite{Ramesh95} showed that the
randomized version of this problem has a competitive ratio at least
$\Omega(k^2/(\log k)^{1+\e})$ for any $\e > 0$; moreover, his
$O(k^{13})$-competitive algorithm was improved to a nearly-tight bound
of $O(k^2 \log \Delta)$ in recent exciting result by Bubeck, Coester,
and Rabani~\cite{https://doi.org/10.48550/arxiv.2202.04551}. 

% If we do not
% bound the number of leaves, then the adversary has no reason to kill
% leaves, and the problem essentially becomes the same as our framework
% (without predictions). \alert{Double-check that we can run our
%   algorithm on any instance of LGT, and hence can use their algorithm
%   to robustify our algos.}

% The paper of 
% \href{https://sites.cs.ucsb.edu/~suri/cs231/Rlist/search.pdf}{Baeza-Yates,
%   Culberson, and Rawlins}. \cite{Baezayates1993234}
% Also see \href{https://www.mit.edu/~jaillet/general/online-or.pdf}{Online
%   Searching} by Jaillet and Stafford.

Kalyanasundaram and Pruhs \cite{Kalyanasundaram1993139} study
the exploration problem (which they call the \emph{searching} problem)
in the geometric setting of $k$ polygonal obstacles with bounded
aspect ratio in the plane.
Some of their results extend to the \emph{mapping} problem, where they must
determine the locations of all obstacles.  Deng and
Papadimitriou~\cite{deng1999exploring} study the mapping problem,
where the goal is to traverse \emph{all edges} of an unknown directed
graph: they give an algorithm with cost $2|E|$ for Eulerian graphs
(whereas $OPT = |E|$), and cost $\exp(O(d \log d)) |E|$ for graphs
with imbalance at most $d$. Deng, Kameda, and
Papadimitriou~\cite{DengKP98} give an algorithm to map two-dimensional
rectilinear, polygonal environments with a bounded number of
obstacles.

Kalyanasundaram and Pruhs~\cite{kalyanasundaram1994constructing}
consider a different version of the mapping problem, where the goal is
to visit all vertices of an unknown graph using a tour of least cost.
They give an algorithm that is $O(1)$-competitive on planar graphs.
Megow et al.~\cite{Megow201262} extend their algorithm to graphs with
bounded genus, and also show limitations of the algorithm from \cite{kalyanasundaram1994constructing}.

Blum, Raghavan and Schieber~\cite{BlumRS97} study the \emph{point-to-point navigation}
problem of finding a minimum-length path between two known
locations $s$ and $t$ in a rectilinear environment; the obstacles are
unknown axis-parallel rectangles. Their $O(\sqrt{n})$-competitiveness
is best possible given the lower bound
in~\cite{PAPADIMITRIOU1991127}. \cite{KarloffRR94} give lower bounds
for randomized algorithms in this setting.

Our work is related in spirit to graph search algorithms like
$A^*$-search which use \emph{score functions} to choose the next leaf
to explore. One line of work giving provably good algorithms is that
of Goldberg and Harrelson~\cite{GoldbergH05} on computing shortest paths without exploring the entire
graph. Another line of work related in spirit to ours is that of Karp, Saks,
and Wigderson~\cite{DBLP:conf/focs/KarpSW86} on branch-and-bound (see
also \cite{DBLP:journals/jacm/KarpZ93}).

\textbf{Noisy Binary Search.}  A very closely related line of work is
that of computing under noisy queries~\cite{FeigeRPU94}. In this
widely-used model, the agent can query nodes: each node ``points'' to
a neighbor that is closer to the goal, though some of these answers
may be incorrect. Some of these works
include~\cite{OnakP06,MozesOW08,Emamjomeh-Zadeh16,DeligkasMS19,DereniowskiTUW19,BoczkowskiFKR21}.
Apart from the difference in the information model (these works
imagine knowing the entire graph) and the nature of hints
(``gradient'' information pointing to a better node, instead of
information about the quality/score of the node), these works often
assume the errors are independent, or adversarial with bounded noise
rate. Most of these works allow random-access to nodes and seek to
minimize the \emph{number} of queries instead of the distance
traveled, though an exception is the work
of~\cite{BoczkowskiFKR21}. % \todo{AG: Read more of these papers!! Refer
% to them earlier?}
As far as we can see, the connections between our models is only in spirit.
Moreover, we show in \S\ref{sec:grad} that results of the
kind we prove are impossible if the predictions don't give us distance
information but instead just edge ``gradients''.
% \alert{AG: check out
%   the paper
%   \href{https://proceedings.neurips.cc/paper/2017/file/3f647cadf56541fb9513cb63ec370187-Paper.pdf}{A
%     general framework for robust interactive learning} to see if it
%   has any interesting applications.}

\textbf{Algorithms with Predictions.}
Our work is related to the exciting line of research on algorithms
with predictions, such as in ad-allocation~\cite{MahdianNS07}, auction
pricing \cite{medina2017revenue}, page
migration~\cite{indyk2020online}, flow
allocation~\cite{lavastida2020learnable}, scheduling
\cite{purohit2018improving, lattanzi2020online,
  mitzenmacher2020scheduling}, frequency estimation
\cite{hsu2019learning}, speed scaling~\cite{BamasMRS20}, Bloom filters
\cite{mitzenmacher2018model}, bipartite matching and secretary
problems~\cite{AntoniadisGKK20,DBLP:conf/sigecom/DuttingLLV21}, and online linear
optimization~\cite{bhaskara2020online}. 

%\todo{AG: Golowich and Moitra?}

%%% Local Variables:
%%% mode: latex
%%% TeX-master: "main"
%%% End:

\section{Problem Setup and Definitions}
\label{sec:definitions}

%\paragraph{The Problem Setup:} 
We consider an underlying graph $G = (V,E)$ with a known root node $r$
and an unknown \emph{goal} node $g$. (For most of this paper, we
consider the unweighted setting where all edge have unit length;
\S\ref{sec:analys-bound-doubl-2} and \S\ref{sec:general-lengths}
discuss cases where edge lengths are positive integers.) Each node has
degree at most $\Delta$. Let $d(u,v)$ denote the distance between
nodes $u,v$ for any $u,v \in V$, and let $D:=d(r,g)$ be the optimal
distance from $r$ to the goal node $g$.
% We primarily focus on
% unit-length edges; see and \S\ref{sec:analys-bound-doubl-2} for
% discussions and results about general edge-lengths.

\begin{wrapfigure}{R}{0.4\textwidth}
  \centering
    \scalebox{1}{    \tikzset{every picture/.style={line width=0.75pt}} %set default line width to 0.75pt        

\begin{tikzpicture}[x=0.75pt,y=0.75pt,yscale=-1,xscale=1]
%uncomment if require: \path (0,213); %set diagram left start at 0, and has height of 213

%Straight Lines [id:da18481582060508317] 
\draw  [dash pattern={on 4.5pt off 4.5pt}]  (183.7,86.3) -- (216.3,100.3) ;
%Straight Lines [id:da6990337136340106] 
\draw  [dash pattern={on 4.5pt off 4.5pt}]  (183.7,86.3) -- (216.3,78.3) ;
%Straight Lines [id:da46242128860235887] 
\draw  [dash pattern={on 4.5pt off 4.5pt}]  (216.3,173.7) -- (248.9,173.7) ;
%Straight Lines [id:da15240667927468965] 
\draw  [dash pattern={on 4.5pt off 4.5pt}]  (183.7,116.3) -- (216.3,116.3) ;
%Straight Lines [id:da2910732160795291] 
\draw  [dash pattern={on 4.5pt off 4.5pt}]  (216.3,143.7) -- (248.6,157.3) ;
%Straight Lines [id:da037579344631267775] 
\draw  [dash pattern={on 4.5pt off 4.5pt}]  (216.3,143.7) -- (236.14,138.54) -- (248.6,135.3) ;
%Straight Lines [id:da9274941302963537] 
\draw    (183.7,156.3) -- (216.3,173.7) ;
%Straight Lines [id:da9023785609401205] 
\draw    (183.7,156.3) -- (216.3,143.7) ;
%Straight Lines [id:da1530617686749025] 
\draw    (146.3,86.3) -- (183.7,86.3) ;
%Straight Lines [id:da6752454549746616] 
\draw    (146.3,143.1) -- (183.7,116.3) ;
%Straight Lines [id:da45557641951574124] 
\draw    (146.3,143.1) -- (183.7,156.3) ;
%Shape: Circle [id:dp6911501179331467] 
\draw  [fill={rgb, 255:red, 208; green, 2; blue, 27 }  ,fill opacity=1 ] (180,156.3) .. controls (180,154.26) and (181.66,152.6) .. (183.7,152.6) .. controls (185.74,152.6) and (187.4,154.26) .. (187.4,156.3) .. controls (187.4,158.34) and (185.74,160) .. (183.7,160) .. controls (181.66,160) and (180,158.34) .. (180,156.3) -- cycle ;
%Shape: Circle [id:dp054970597087039774] 
\draw  [fill={rgb, 255:red, 208; green, 2; blue, 27 }  ,fill opacity=1 ] (180,116.3) .. controls (180,114.26) and (181.66,112.6) .. (183.7,112.6) .. controls (185.74,112.6) and (187.4,114.26) .. (187.4,116.3) .. controls (187.4,118.34) and (185.74,120) .. (183.7,120) .. controls (181.66,120) and (180,118.34) .. (180,116.3) -- cycle ;
%Straight Lines [id:da8295761843382856] 
\draw    (116.3,115.7) -- (146.3,86.3) ;
%Straight Lines [id:da1953896173671943] 
\draw    (116.3,115.7) -- (146.3,143.7) ;
%Shape: Circle [id:dp07075917174939961] 
\draw  [fill={rgb, 255:red, 208; green, 2; blue, 27 }  ,fill opacity=1 ] (112.6,115.7) .. controls (112.6,113.66) and (114.26,112) .. (116.3,112) .. controls (118.34,112) and (120,113.66) .. (120,115.7) .. controls (120,117.74) and (118.34,119.4) .. (116.3,119.4) .. controls (114.26,119.4) and (112.6,117.74) .. (112.6,115.7) -- cycle ;
%Shape: Circle [id:dp3413225968701201] 
\draw  [fill={rgb, 255:red, 208; green, 2; blue, 27 }  ,fill opacity=1 ] (142.6,143.7) .. controls (142.6,141.66) and (144.26,140) .. (146.3,140) .. controls (148.34,140) and (150,141.66) .. (150,143.7) .. controls (150,145.74) and (148.34,147.4) .. (146.3,147.4) .. controls (144.26,147.4) and (142.6,145.74) .. (142.6,143.7) -- cycle ;
%Shape: Circle [id:dp09493116903630794] 
\draw  [fill={rgb, 255:red, 208; green, 2; blue, 27 }  ,fill opacity=1 ] (142.6,86.3) .. controls (142.6,84.26) and (144.26,82.6) .. (146.3,82.6) .. controls (148.34,82.6) and (150,84.26) .. (150,86.3) .. controls (150,88.34) and (148.34,90) .. (146.3,90) .. controls (144.26,90) and (142.6,88.34) .. (142.6,86.3) -- cycle ;
%Shape: Circle [id:dp28178680271693524] 
\draw  [fill={rgb, 255:red, 208; green, 2; blue, 27 }  ,fill opacity=1 ] (180,86.3) .. controls (180,84.26) and (181.66,82.6) .. (183.7,82.6) .. controls (185.74,82.6) and (187.4,84.26) .. (187.4,86.3) .. controls (187.4,88.34) and (185.74,90) .. (183.7,90) .. controls (181.66,90) and (180,88.34) .. (180,86.3) -- cycle ;
%Shape: Circle [id:dp7456628002940937] 
\draw  [fill={rgb, 255:red, 208; green, 2; blue, 27 }  ,fill opacity=1 ] (212.6,173.7) .. controls (212.6,171.66) and (214.26,170) .. (216.3,170) .. controls (218.34,170) and (220,171.66) .. (220,173.7) .. controls (220,175.74) and (218.34,177.4) .. (216.3,177.4) .. controls (214.26,177.4) and (212.6,175.74) .. (212.6,173.7) -- cycle ;
%Shape: Circle [id:dp4726514286580503] 
\draw  [fill={rgb, 255:red, 208; green, 2; blue, 27 }  ,fill opacity=1 ] (212.6,143.7) .. controls (212.6,141.66) and (214.26,140) .. (216.3,140) .. controls (218.34,140) and (220,141.66) .. (220,143.7) .. controls (220,145.74) and (218.34,147.4) .. (216.3,147.4) .. controls (214.26,147.4) and (212.6,145.74) .. (212.6,143.7) -- cycle ;
%Shape: Circle [id:dp5726600086888245] 
\draw  [fill={rgb, 255:red, 74; green, 144; blue, 226 }  ,fill opacity=1 ] (244.9,157.3) .. controls (244.9,155.26) and (246.56,153.6) .. (248.6,153.6) .. controls (250.64,153.6) and (252.3,155.26) .. (252.3,157.3) .. controls (252.3,159.34) and (250.64,161) .. (248.6,161) .. controls (246.56,161) and (244.9,159.34) .. (244.9,157.3) -- cycle ;
%Shape: Circle [id:dp09250259296139052] 
\draw  [fill={rgb, 255:red, 74; green, 144; blue, 226 }  ,fill opacity=1 ] (244.9,135.3) .. controls (244.9,133.26) and (246.56,131.6) .. (248.6,131.6) .. controls (250.64,131.6) and (252.3,133.26) .. (252.3,135.3) .. controls (252.3,137.34) and (250.64,139) .. (248.6,139) .. controls (246.56,139) and (244.9,137.34) .. (244.9,135.3) -- cycle ;
%Shape: Circle [id:dp09378774334229645] 
\draw  [fill={rgb, 255:red, 74; green, 144; blue, 226 }  ,fill opacity=1 ] (212.6,116.3) .. controls (212.6,114.26) and (214.26,112.6) .. (216.3,112.6) .. controls (218.34,112.6) and (220,114.26) .. (220,116.3) .. controls (220,118.34) and (218.34,120) .. (216.3,120) .. controls (214.26,120) and (212.6,118.34) .. (212.6,116.3) -- cycle ;
%Shape: Circle [id:dp9700992069923229] 
\draw  [fill={rgb, 255:red, 74; green, 144; blue, 226 }  ,fill opacity=1 ] (245.2,173.7) .. controls (245.2,171.66) and (246.86,170) .. (248.9,170) .. controls (250.94,170) and (252.6,171.66) .. (252.6,173.7) .. controls (252.6,175.74) and (250.94,177.4) .. (248.9,177.4) .. controls (246.86,177.4) and (245.2,175.74) .. (245.2,173.7) -- cycle ;
%Shape: Circle [id:dp9261689247851652] 
\draw  [fill={rgb, 255:red, 74; green, 144; blue, 226 }  ,fill opacity=1 ] (212.6,100.3) .. controls (212.6,98.26) and (214.26,96.6) .. (216.3,96.6) .. controls (218.34,96.6) and (220,98.26) .. (220,100.3) .. controls (220,102.34) and (218.34,104) .. (216.3,104) .. controls (214.26,104) and (212.6,102.34) .. (212.6,100.3) -- cycle ;
%Shape: Circle [id:dp9028644498377398] 
\draw  [fill={rgb, 255:red, 74; green, 144; blue, 226 }  ,fill opacity=1 ] (212.6,78.3) .. controls (212.6,76.26) and (214.26,74.6) .. (216.3,74.6) .. controls (218.34,74.6) and (220,76.26) .. (220,78.3) .. controls (220,80.34) and (218.34,82) .. (216.3,82) .. controls (214.26,82) and (212.6,80.34) .. (212.6,78.3) -- cycle ;
%Shape: Circle [id:dp9548585406935866] 
\draw  [fill={rgb, 255:red, 208; green, 2; blue, 27 }  ,fill opacity=1 ] (58,144.3) .. controls (58,142.26) and (59.66,140.6) .. (61.7,140.6) .. controls (63.74,140.6) and (65.4,142.26) .. (65.4,144.3) .. controls (65.4,146.34) and (63.74,148) .. (61.7,148) .. controls (59.66,148) and (58,146.34) .. (58,144.3) -- cycle ;
%Shape: Circle [id:dp17802253307570437] 
\draw  [fill={rgb, 255:red, 74; green, 144; blue, 226 }  ,fill opacity=1 ] (57.6,164.7) .. controls (57.6,162.66) and (59.26,161) .. (61.3,161) .. controls (63.34,161) and (65,162.66) .. (65,164.7) .. controls (65,166.74) and (63.34,168.4) .. (61.3,168.4) .. controls (59.26,168.4) and (57.6,166.74) .. (57.6,164.7) -- cycle ;
%Shape: Rectangle [id:dp8668076956231014] 
\draw  [color={rgb, 255:red, 0; green, 0; blue, 0 }  ,draw opacity=0.5 ] (51,133.67) -- (109,133.67) -- (109,176) -- (51,176) -- cycle ;

% Text Node
\draw (69,135.4) node [anchor=north west][inner sep=0.75pt]    {${\textstyle V_{t}}$};
% Text Node
\draw (69,156.4) node [anchor=north west][inner sep=0.75pt]    {${\textstyle \partial V_{t}}$};

\end{tikzpicture}}
    \caption{The observed vertices $V_t \cup \partial V_t$ (and extended subtree $\xT^t := T[V_t \cup \partial V_t]$) at some intermediate stage of the algorithm. Visited nodes $V_t$ are shown in red, and their un-visited neighbors $\partial V_t$ in blue.}
    \label{fig: V and Partial V}
\end{wrapfigure}
Let us formally define the problem setup.
An agent initially starts at root $r$, and wants to visit goal $g$ while traversing the minimum number of edges.
In each timestep $t \in \{1, 2, \ldots\}$, the agent moves from some node $v_{t-1}$ to node $v_t$.
%; upon visiting a node for the first time, the agent becomes aware of the nodes in its neighborhood. 
We denote the \emph{visited vertices} at the start of round $t$ by $V_{t-1}$ (with $V_0 = \{r\}$), and use $\partial V_{t-1}$ to denote the \emph{neighboring vertices}---those not in $V_{t-1}$ but having at least one neighbor in
$V_{t-1}$. Their union $V_{t-1} \cup \partial V_{t-1}$ is the set of \emph{observed vertices} at the end of time $t-1$.
Each time $t$ the agent \emph{visits} a new node $v_t$ such that
$V_t := V_{t-1} \cup \{v_t\}$, and it pays the movement cost
$d(v_{t-1}, v_t)$, where $v_0 = r$. Finally, when $v_t = g$ and the
agent has reached the goal, the algorithm stops. The identity of the
goal vertex is known when---and only when---the agent visits it, and
we let $\tau^*$ denote this timestep.
% (Changing the model so that the
% agent knows the goal identity once it is observed 
Our aim is to design an algorithm that reaches the goal state with minimum total movement cost:
\[ \sum_{t = 1}^{\tau^*} d^{t-1}(v_{t-1}, v_t). \]
Within the above setting, we consider two problems:
\begin{itemize}
\item In the \emph{planning} problem, the agent knows the graph $G$
  (though not the goal $g$), and in addition, is given a
  \emph{prediction} $f(v) \in \ZZ$ for each $v \in V$ of its distance to the goal $g$; it can then use this information to plan its search trajectory.
\item In the \emph{exploration} problem, the graph $G$ and the
  predictions $f(v) \in \ZZ$ are initially unknown to the agent, and these are revealed only via exploration; in particular, upon visiting a node for the first time, the agent becomes aware of previously unobserved nodes in $v$'s neighborhood.
Thus, at the end of timestep $t$, the agent knows the set of visited vertices $V_t$, neighboring vertices $\partial V_t$, and the predictions $f(v)$ for each observed vertex
$v \in V_t \cup \partial V_t$. 
\end{itemize}
In both cases, we define $\cE := \{v \in V \mid f(v)\ne d(g,v)\}$ to be the
set of \emph{erroneous nodes}. % , and hence $\error = |\cE|$.
Extending this notation, for the exploration problem, we define $\cE^t := \cE \cap V_t$ as the erroneous nodes visited
by time $t$. %, and define $\error^t = |\cE^t|$.

% For the planning problem, we can pre-process the predictions $f(v)$ to compute the
% \emph{implied-error} $\varphi(v)$,
%  which for any node $v$
% gives the error if the goal $g=v$. Formally,
% \begin{definition}[Implied-error]
%   The implied-error $\varphi: V \to \ZZ$ maps each node $v\in V$ to
%   $\varphi(v):=|\{u \in V \mid d(u,v)\neq f(u)\}|$, which is the $\ell_0$ error if the
%   goal were at $v$.
% \end{definition}

%%% Local Variables:
%%% mode: latex
%%% TeX-master: "main"
%%% End:

\section{Exploring with a Known Target Distance}
\label{sec:known-D}

Recall that our algorithm for the exploration problem on trees proceeds via the \emph{known-distance} version of the problem: in
addition to seeing the predictions at the various nodes as we explore
the tree, we are promised that the \emph{distance from the starting node/root
$r$ to the goal state $g$ is is exactly some value $D$}, i.e.,
$d(r,g) = D$.  The main result of this section is \Cref{thm:known-D},
and we restate a rigorous version here.

\begin{theorem}
  \label{thm:Known-D rigorous}
  If $D=d(r,g)$, the algorithm $\AlgoKnownD(r, D,+\infty)$ finds the goal node $g$ incurring a cost of at most $d(r,g) +
  O(\Delta |\cE|)$.
\end{theorem}

Algorithm $\AlgoKnownD$ is stated in \Cref{alg:degree-k}. For better understanding of it, we first give some key definitions.

% that we can find the goal with cost at most the
% distance plus a linear additional loss that depends on the number of
% erroneous nodes that we see.

\subsection{Definitions: Anchors, Degeneracy, and Criticality}
\label{sec:known-D-defs}
For an unweighted tree $T$, we define the \emph{level} of node $v$
with respect to the root $r$ to be $\ell(v):= d(r,v)$, and so \emph{level
  $L$} denotes the set of nodes $v$ such that $d(r,v) = \ell(v) =
L$. Since the tree is rooted, there are clearly defined notions of
parent and child, ancestor and descendent.
% Given a pair of
% nodes $(u,v)$ directly connected by an edge, if $\ell(v) = \ell(u)+1$, then
% we refer to $v$ as $u$'s \emph{child} (and $u$ as $v$'s
% \emph{parent}). Similarly, a node $u$ is $v$'s \emph{ancestor} (and
% $v$ is $u$'s \emph{descendant}) if there exist a sequence of nodes
% $a_0,a_1,a_2,\ldots,a_p$ such that $a_0=v,a_p=u$, and for any
% $i=1,2,\ldots,p$, $a_i$ is $a_{i-1}$'s parent. 
Each node is
both an ancestor and a descendant of itself. For any node $v$, let
$T_v$ denote the \emph{subtree} rooted at $v$.  Extending this
notation, we define the \emph{visited subtree} $\eT^t := T[V_t]$, and
the \emph{extended subtree} $\xT^t := T[V_t \cup \partial V_t]$, and
let $\eT^t_v$ and $\xT^t_v$ be the subtrees of $\eT^t$ and $\xT^t$
rooted at $v$.

\begin{definition}[Active and Degenerate nodes]
\label{definition: k-degree degenerate}
In the exploration setting, at the end of timestep $t$, a node $v\in V_t\cup\partial V_t$ is
\emph{active} if $\eT^t_v \neq \xT^t_{v}$, i.e., there are observed descendants of $v$ (including itself) that have not been visited.\\ 
An active node $v\in V_t\cup\partial V_t$ is \emph{degenerate} at the end of timestep $t$ if 
% \alert{
it has a unique child node in $\xT^t$ that is active.
% }
% \todo{SB: Zhouzi, can you check the alert secions and figure now}\balert{Zhouzi: Yes! This is exactly what I am thinking of. Thank you! }
  %there is a unique index $j$ such that $\chi_j(v)$ is active.
\end{definition}
%\balert{Zhouzi: Actually, according to the definition (and also what's in my mind), I think  the blue nodes are also active. Thus the yellow node that has only a blue child should be red. I think the algorithm works in this definition.}

In other words, all nodes which have un-visited descendants (including those in the frontier $\partial V_t$) are active.
Active nodes are further partitioned into degenerate nodes that have exactly one child subtree that has not been fully visited, and active nodes that have at least two active children. See~\cref{fig: Anchor} for an illustration.

A crucial definition for our algorithms is that of \emph{anchor} nodes: 
\begin{definition}[Anchor]
  \label{def:anchor}
  For node $u\in T$, define its \emph{anchor} $\tau(u)$ to be its
  ancestor in level $\alpha(u) := \frac12(D+\ell(u)-f(u))$. If the value
  $\alpha(u)$ is negative, or is not an integer, or node $u$ itself belongs
  at level smaller than $\alpha(u)$, we say that $u$ has no anchor
  and that $\tau(u) = \bot$.
\end{definition}
\cref{fig: Anchor} demonstrates the location of an anchor node $\tau(u)$ for given node $u$; it also illustrates the following claim, which forms the main rationale behind the definition: 
\begin{claim}
  \label{clm:anchor}
  If the prediction for some node $u$ is correct, then its anchor
  $\tau(u)$ is the least common ancestor (in terms of level $\ell$) of $u$ and the goal $g$.
  Consequently, if a node $u$ has no anchor, or if its anchor does not lie on the path $P^*$ from $r$ to $g$, then $u \in \cE$.
\end{claim}

\begin{figure}[t]
    \centering
    \scalebox{1}{    \input{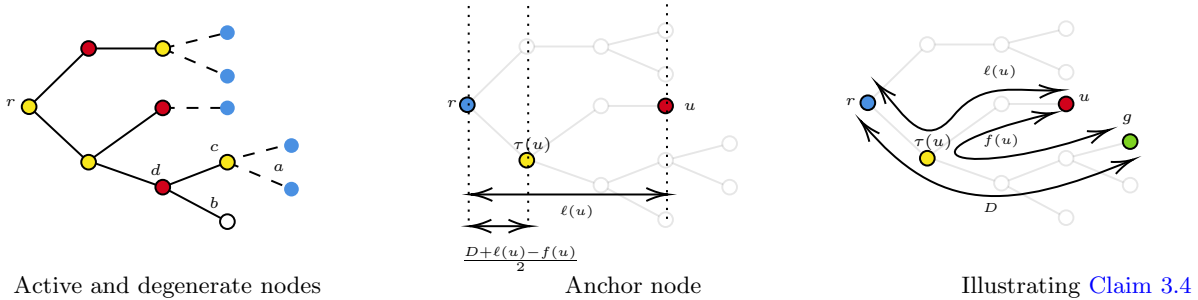}}
    \caption{The first figure from the left illustrates active and degenerate nodes. Nodes such as $a$ (shaded in blue) are in $\partial V_t$ while the rest are visited nodes in $V_t$.
    % \alert{
    Unshaded node $b$ is inactive (since it has no un-visited descendant), while all other shaded nodes (blue, yellow and red) are active.
    Among the active nodes, nodes such as $c$ (shaded in yellow) are non-degenerate nodes as they have at least two active children. Finally nodes such as $d$ (shaded in red) are degenerate as they have exactly one active child.
    % }
    \\ 
    The second and third figures give an example of anchor node $\tau(u)$ (in yellow) at level $\frac12(D+\ell(u)-f(u))$ for given node $u$ (in red) at level $\ell(u)$. The rightmost figure (with root $r$ and goal $g$ also indicated) illustrates~\cref{clm:anchor}, showing that when $u$'s prediction $f(u)$ is correct, then its anchor is the least common ancestor of $u$ and goal $g$ (since $D+\ell(u)-f(u)$ is equal to twice the distance of $\tau(u)$ from $r$).}
    \label{fig: Anchor}
\end{figure}

For any node $v\in T$, define its children be $\chi_i(v)$ for
$i=1,2,\ldots,\Delta_v$, where $\Delta_v \leq \Delta$ is the number of
children for $v$. Note that the order is arbitrary but prescribed and
fixed throughout the algorithm. For any time $t$, node $v$, and $i \in
[\Delta_v]$, define the visited portion of the subtree rooted at the $i^{th}$ child as $\ST^t_i(v):= \eT^t_{\chi_i(v)}$.

\begin{definition}[Loads $\sigma_i$ and $\sigma$]
\label{def:loads}
  For any time $t$, node $v$, and index $i \in [\Delta_v]$, define  \[ \sigma_i^t(v):=|\{u\in
  \ST_i^t(v) \mid \tau(u)=v\}|. \] 
In other words, $\sigma_i^t(v)$ is the number of nodes in $\ST^t_i(v)$
that have $v$ as their anchor.  Define
$\sigma^t(v)=\sum_{i=1}^{\Delta_v} \sigma_i^t(v)$ to be the total
number of nodes in $\eT^t_v \setminus \{v\}$ which have $v$ as their
anchor.
\end{definition}

\begin{definition}[Critical Node]
  \label{definition: k-degree critical}
  For any time $t$, active and non-degenerate node $v$, and index $j \in [\Delta_v]$, let
  $q_j :=\arg\min_{i \neq j} \{ \sigma_i^t(v) \mid \text{$\chi_i(v)$
    is active at time $t$} \}$. Call $v$ a \emph{critical node with respect to $j$ at time $t$} if it satisfies
  \begin{enumerate}[(i),nolistsep]
  % \item $v$ is not degenerate, 
  \item $\sigma^t_j(v)\geq 2\sigma^t_{q_j}(v)$, namely, the number of
    nodes of $\ST^t_j(v)$ that have $v$ as their anchor is at least
    twice larger than the number of nodes of $\ST^t_{q_j}(v)$ that
    have $v$ as their anchor; and
  \item $2\sigma^t_j(v)\geq |\ST^t_j(v)|$, namely,
  at least half of the nodes of $\ST^t_j(v)$ have  $v$ as their
    anchor.
  \end{enumerate}
\end{definition}
% \alert{Conditions (1) and (2) are true for any node of
% $\ST^t_j(v)$ (not only $u$). So $v$ is critical w.r.t.
% to all the nodes in $\ST^t_j(v)$, right?}

\subsection{The \AlgoKnownD Algorithm}
 
Equipped with the definitions in \S\ref{sec:known-D-defs}, 
at a high level, the main idea of the algorithm is to balance the
\emph{loads} (as defined in \Cref{def:loads}) of all the nodes
$v$. Note that if the goal $g\in \ST_i(v)$, then the nodes
$u\in \ST_i(v)$ that have $v$ as their anchor (i.e., $\tau(u)=v$) have
erroneous predictions; hence balancing the loads automatically balances
the cost and the budget.  To balance the loads, we use the definition
of a \emph{critical} node (see \Cref{definition: k-degree critical}):
whenever a node $v$ becomes critical, the algorithm goes back and
explores another subtree of $v$, thereby maintaining the balance.

More precisely,
our
algorithm \AlgoKnownD does the
following: % \alert{explain the steps in words}
at each time step $t$, it checks whether there is a node that is critical.
%active, non-degenerate and critical. 
If there is no such node, the
algorithm performs one more step of the current DFS, giving priority to the unexplored child of $v_t$ with smallest prediction. On the other hand, if there is a
critical node $v$, then this $v$ must be the anchor $\tau(v_t)$. In
this case the algorithm pauses the current DFS, returns to the anchor
$\tau(v_t)$ and resumes the DFS in $\tau(v_t)$'s child subtree having
the smallest load (say $\ST_q(\tau(v_t))$). This DFS may have been paused at some time $t'<t$,
and hence is continued starting at node $v_{t'}$. The variable $\mem(v)$ saves the vertex that the algorithm left the subtree rooted on $v$ last time. For example, in this case $\mem(\chi_q(\tau(v_t)))=v_{t'}$. If no such time
$t'$ exists, the algorithm starts a new DFS from some child of
$\tau(v_t)$ whose subtree has the smallest load (in this case, $\mem(\chi_q(\tau(v_t)))=\bot$). The pseudocode
appears as \Cref{alg:degree-k}.

% \alert{Change the algorithm to have a stopping time $B$}.

\begin{algorithm}[h]
  \caption{$\AlgoKnownD(r, D, B)$}\label{alg:degree-k}
  $v_0 \gets r$, $t \gets 0$\;
  $\mem(r)\gets r$ and $\mem(v)\gets \bot$ for all $v \neq r$\;
  \While{$v_t \neq g$ and $|V_t| < B$}{
    \If{$\tau(v_t) \neq \bot$ and $\tau(v_t)$ is active and not degenerate and
      $\tau(v_t)$ is
      critical w.r.t.\ the index of the subtree containing $v_t$ at time $t$ %(\alert{Here not degenerate
      % may be included in the definition of
      % critical})
    }{
      \label{l:check} 
      $q\gets$ the child index $q$ s.t.  $\tau(v_t)$ is critical w.r.t. $q$\;
      \leIf{$\mem(\chi_q(\tau(v_t)))=\bot$}{$v_{t+1} \gets \chi_q(\tau(v_t)$}
      {$u\gets \mem(\chi_q(\tau(v_t))$}
    }
    \Else{
      $u\gets v_t$
    }
    \While{$v_{t+1}$ undefined and $u$ has no child}{
      \label{iterate: in Algorithm2}
      $w\gets$ $u$'s closest active ancestor  \;
    %   \balert{
      $q \gets \arg\min_i \{ \sigma_i^t(w) \mid \chi_i(w)
        \text{ active }\}$
        % }
        \;
      \leIf{$\mem(\chi_q(w))=\bot$}{$v_{t+1}\gets \chi_q(w)$}{$u\gets \mem(\chi_q(w))$}
    }
    \lIf{$v_{t+1}$ undefined}{$v_{t+1}\gets$ $u$'s child with smallest prediction}
    \lForEach{ancestor $u$ of $v_{t+1}$}{
      $\mem(u)\gets v_{t+1}$
    }
    $t \gets t+1$\;
  }
\end{algorithm}

A few observations: (a) While $D = d(r,g)$ does not appear explicitly in the algorithm, it is used in the definition of \emph{anchors}
(recall \Cref{def:anchor}). Even when $d(r,g)$, the predicted distance at the root, is not the true distance to the goal (as may happen
in \Cref{sec: alg for unknown D}), given any input $D$ in \Cref{alg:degree-k}, we will still define $\tau(v)$ to be
$v$'s ancestor at level $\alpha(u) := \frac12(D+\ell(u)-f(u))$. (b) The new node $v_t$ is
always on the \emph{frontier}: i.e., the nodes which are either leaves
of $T$ or have unvisited children. Moreover, (c) the memory value
$\mem(v) = \bot$ if and only if $v \not\in V_t$, else $\mem(v)$ is on the
frontier in the subtree below $v$.

% \alert{Add one figure to illustrate how Algorithm \AlgoKnownD runs if possible.}

\subsection{Analysis for the \AlgoKnownD Algorithm}
\label{sec:analysis for algorithm 2}

The proof of \Cref{thm:Known-D rigorous} % that our total exploration cost is
% $D + O(\Delta \cdot |\cE|)$ 
proceeds in two steps. The first step is
to show that the total amount of ``extra'' exploration, i.e., the
number of nodes that do not lie on the $r$-$g$ path, is
$O(\Delta \cdot |\cE|)$. Formally, let $P^*$ denote the $r$-$g$ path;
for the rest of this section, suppose $g\in \ST_1(v)$ for all
$v\in P^*$. Define the \emph{extra exploration} to be the number of
nodes visited in the subtrees hanging off this path: 
\[ \ExEx(t) := \sum_{v\in P^*} \sum_{i\neq 1} |\ST_i^t(v)|. \]

\begin{lemma}[Bounded Extra Exploration]
  \label{lemma: extra exploration totally bounded}
  For all times $t^*$, 
  $\ExEx(t^*) \leq 7\Delta \cdot
    |\cE^{t^*}|$. 
\end{lemma}

Next,  we need to control the total distance traveled, which is the 
second step of our analysis:
\begin{lemma}[Bounded Cost]
  \label{lem:distance-by-extra-exploration}
  For all times $t^*$,
  \[ \sum_{t\leq t^*} d(v_{t-1}, v_t) \leq d(r,v_{t^*}) +  10 \ExEx(t^*)  +
    16 |\cE^{t^*}|. \]
%   \alert{This theorem is proved for the time $t^*$ when the algorithm
%     stops. But we use it later for general times $t$, so change.}
%   Hence at time $t$, this quantity is at most $19t$.
\end{lemma}
Using the lemmas above (setting $t^*$ to be the time $\tau^*$ when we
reach the goal) proves \Cref{thm:known-D}. In the following sections,
we now prove \Cref{lemma: extra exploration totally bounded,lem:distance-by-extra-exploration}.

\subsection{Bounding the Extra Exploration}
\label{sec: bounded extra exploration alg 2}

\begin{lemma}
  \label{lemma: error for each node alg2}
  For any node $v\in T^t$, define $x^t(v)$ as follows:
  \begin{enumerate}[nolistsep]
  \item if $g\notin T_v$, then $x^t(v):=\sigma^t(v)$.
  \item if $g\in T_v \setminus \{v\}$, let $g\in T_{\chi_j(v)}$. Define $y_1^t(v):=\sigma_j^t(v)$, $y_2^t(v):=\sum_{i\neq j} (|\ST_i^t(v)|-\sigma_i^t(v))$ and
    $x^t(v):= y_1^t(v) + y_2^t(v)$.
  \end{enumerate}
  Then $\sum_{v\in T^t}x^t(v)\leq 2|\cE^t|$.
\end{lemma}

\begin{proof}
  Let $P^*$ be the $r$-$g$ path in $T$. If $g\notin T_v$ (i.e.,
  $v\notin P^*$), then by \Cref{clm:anchor} all the nodes with $v$ as
  anchor belong to $\cE$. Else suppose $g\in T_v$ (i.e., $v\in P^*$),
  and suppose $g\in T_{\chi_j(v)}$. Now all nodes $u$ in $\ST_j(v)$
  having anchor $v$ belong to $\cE$, since the least common ancestor
  of $u$ and $g$ can be no higher than $\chi_j(v)$. This means
  \[\sum_{v\in T^t\setminus P^*}x^t(v) + \sum_{v\in P^*}y_1^t(v) \leq
    \sum_{v\in T^t} |\{u\in \cE \mid \tau(u)=v\}|\leq |\cE^t|.\]
  Finally, suppose $g\in T_v$ (i.e., $v\in P^*$) and
  $g\in T_{\chi_j(v)}$. Now for any node $u \in T_{\chi_i(v)}$ for
  $i\neq j$, the least common ancestor of $u$ and $g$ is $v$. Hence
  nodes in $T_{\chi_i(v)}$ for $i\neq j$ whose anchor is not $v$ must
  be wrongly predicted. Denote the set of such nodes by $Y_2^t(v)$.
  Note that $|Y_2^t(v)|=y_2^t(v)$, and $Y_2^t(v)$ for each $v \in P^*$
  are disjoint. Hence we have
  \[\sum_{v\in P^*}y_2^t(v)\leq \sum_{v\in P^*}|Y_2^t(v)|\leq |\cE^t|.\]
  Summing the two inequalities we get the proof.
\end{proof}

\begin{lemma}
  \label{lemma: not entering if sigma is not minimal}
  For any node $v\in T$ and any index $i\in \{1,2, \ldots ,\Delta_v\}$
  such that
  $\sigma_i^t(v) > \min_{q} \{\sigma_q^t(v)\mid \chi_q(v) \text{ is
    active at time } t \}$. If $v_t \in T_{\chi_j(v)}$
  for some  $j\neq i$ then $v_{t+1}\notin T_{\chi_i(v)}$.
\end{lemma}
\begin{proof}
  The proof is by contradiction. Assume there is such a time $t$, and
  let
  $w := \arg\min_q \{\sigma_q^t(v)\mid \chi_q(v) \text{ is active at
    time } t\}$. Since $v_{t+1} \in T_{\chi_i(v)}$, the subtree under
  node $\chi_i(v)$ was not fully visited at time $r$ and hence
  $\chi_i(v)$ was active. By the definition of $w$ and the condition
  on $i$ in the lemma statement, we have
  $\sigma_i^t(v)>\sigma_w^t(v)$. Now \Cref{alg:degree-k} will ensure
  that $v_{t+1}$ either remains in $T_{\chi_j(v)}$ or moves into $T_{\chi_w(v)}$.
\end{proof}

\begin{lemma}
  \label{lemma: inequalities guarantee in Alg 2}
  For any node $v$ on the $r$-$g$ path $P^*$, recall the assumption
  that $g\in \ST_1(v)$. For any time $t$ and any $i\neq 1$, at least
  one of the following statements must hold:
  \begin{enumerate}[(I), noitemsep]
  \item $\sigma_i^{t}(v) \leq 2\sigma_1^{t}(v) $.
  \item $2\sigma_i^{t}(v) \leq |\ST_i^{t}(v)| $.
  \item $\sigma_i^t(v)=|\ST_i^t(v)|=1, \sigma_1^t(v)=0$.
  \end{enumerate}
\end{lemma}

\begin{proof}
  For sake of a contradiction, assume there exists $t,i$ such that at
  time $t$ none of the three statements are true, and this is the
  first such time. % The first two statements being false means
  % \[\sigma_i^{t}(v)>2\sigma_1^{t}(v) \qquad \text{and} \qquad
  %   2\sigma_i^{t}(v)>|\ST_i^{t}(v)|.\]
  If $|\ST_i^{t}(v)|=1$, then the falsity of second statement gives
  $\sigma_i^{t}(v) > \nf12\, |\ST_i^{t}(v)|=\nf12$, and so
  $\sigma_i^{t}(v)=1$. Then the first statement being false implies
  $\sigma_1^{t}(v)<\nf12$, which means the third
  statement must hold.

  Henceforth let us assume $|\ST_i^{t}(v)|\geq 2$. Let $t'<t$ be the
  latest time such $v_{t'}\in \ST_i(v)$ and $\tau(v_{t'})=v$. Because the second statement is false, 
  $\sigma_i^{t}(v)> \nf12 \, |\ST_i^{t}(v)| \geq 1$, and so such a time
  $t'$ exists.

 Since $t'$ is the latest time satisfying the condition, we have $\sigma_i^{t}(v)\leq \sigma_i^{t'}(v)+1$. Moreover,  the number of nodes in $C_i^t(v)$ whose anchor is not $v$ does not decrease, hence $|\ST_i^{t}(v)| - \sigma_i^{t}(v)\geq
  |\ST_i^{t'}(v)|-\sigma_i^{t'}(v)$. Also, the number of nodes in $C_1^t(v)$ whose anchor is $v$ does not decrease, hence $\sigma_1^{t}(v)\geq
  \sigma_1^{t'}(v)$. 

  Thus we can get 

\begin{equation}
\label{equation: inequalities for t' alg 2}
    \begin{aligned}
        &\sigma_i^{t'}(v)-2\sigma_1^{t'}(v) \geq \sigma_i^{t}(v)-2\sigma_1^{t}(v) -1 \geq 0\\
        &2\sigma_i^{t'}(v) - |\ST_i^{t'}(v)| \geq 2\sigma_i^{t}(v) - |\ST_i^{t}(v)| - 1\geq 0
    \end{aligned}
\end{equation}

Now if $C_i^{t'}(v)$ is completely visited, then obviously $v_{t'+1}\notin C_i(v)$. Otherwise, $C_i^{t'}(v)$ is active. Also because $g\in \ST_1(v)$, hence $\ST_1(v)$ cannot be completely visited unless the algorithm ends, which means $v$ is not degenerate and $C_1^{t'}(v)$ is still active. Furthermore, we have inequalities \eqref{equation: inequalities for t' alg 2}, hence $v$ must be critical w.r.t. the subtree containing $v_{t'}$ (because taking $q=1$ we get the two inequalities for critical hold, although $\sigma_1^{t'}(v)$ may not be the smallest one). Hence at time $t'+1$ the algorithm will go to a node which is not in $\ST_i(v)$. 

\textbf{If $v_t \notin C_i^t(v)$}:
Note that one of the three statements holds for $t'$.
If one of the first two statements is true to $t'$, then the same statement is also true to $t$ because $\sigma_i^t(v)=\sigma_i^{t'}(v)$, $|C_i^t(v)|=|C_i^{t'}(v)|$ and $\sigma_1^t(v)\geq \sigma_1^{t'}(v)$. Otherwise we have $\sigma_i^t(v) = \sigma_i^{t'}(v) = |C_i^t(v)|=|C_i^{t'}(v)| = 1$. Then if $\sigma_1^t(v)=0$, then the third statement is true to $t$; if $\sigma_1^t(v)\geq 1$, then the first statement is true to $t$.

\textbf{Otherwise $v_t\in C_i^t(v)$}:
By \Cref{lemma: not entering if sigma is not minimal}, there must exist a time $t>t''>t'$ such that $\sigma_1^{t''}(v)\geq \sigma_i^{t''}(v)$ (otherwise the algorithm will never enter $C_i(v)$ since $C_1(v)$ is always active). 
 Hence by the analysis before, we have $\sigma_1^{t''}(v)\geq \sigma_i^{t'}(v)\geq 1$. 
 Because $t'$ is defined as the latest time before $t$ when $v_t\in C_i(v)$, we have $\sigma_i^{t''}(v) = \sigma_i^{t'}(v)$.
 Hence $\sigma_i^{t}(v)\leq \sigma_i^{t'}(v)+1\leq 2\sigma_i^{t''}(v)\leq 2\sigma_1^{t''}(v)\leq 2\sigma_1^t(v)$, which is the first statement in this lemma. 
\end{proof}

\begin{lemma}
  \label{lemma: extra exploration bounded for each node alg2}
  For any node $v$ on the $r$-$g$ path $P^*$, and any time $t$,
  \begin{enumerate}[(i), nolistsep]
  \item if $f(\chi_i(v))=d(\chi_i(v),g)$ for all $i\in [\Delta_v]$
    then $\sum_{i\neq 1}|\ST_i^t(v)|\leq 3\Delta x^t(v)$, 
  \item else $\sum_{i\neq 1}|\ST_i^t(v)|\leq 3\Delta x^t(v)+\Delta $.
  \end{enumerate}
\end{lemma}

\begin{proof} 
  For the first case: if $f(\chi_i(v))=d(\chi_i(v),g)$ for all $i$, then $f(\chi_1(v))$ is the smallest label among all $f(\chi_i(v))$ since the predictions are all correct. Hence by the algorithm, the first node reached among $\{\chi_i(v)\}$ must be $\chi_1(v)$, which means the third statement in \Cref{lemma: inequalities guarantee in Alg 2} cannot hold. By \Cref{lemma: inequalities guarantee in Alg 2}, for any $i,t$, $\sigma_i^{t}(v)\leq 2\sigma_1^t(v)$ or $2\sigma_i^t(v)\leq |\ST_i^t(v)|$.
  
If $\sigma_i^{t}(v)\leq 2\sigma_1^t(v)$: $|\ST_i^t(v)|-\sigma_i^t(v)+\sigma_1^t(v)\geq \sigma_1^t(v)\geq  \sigma_i^t(v)/2$; If $2\sigma_i^t(v)\leq |\ST_i^t(v)|$: $|\ST_i^t(v)|-\sigma_i^t(v)+\sigma_1^t(v)\geq |\ST_i^t(v)|-\sigma_i^t(v) \geq  \sigma_i^t(v)$. Either of them can lead to a conclusion that \[|\ST_i^t(v)|-\sigma_i^t(v)+\sigma_1^t(v)\geq \sigma_i^t(v)/2.\] 
Denote $x_i^t(v):=|\ST_i^t(v)|-\sigma_i^t(v)+\sigma_1^t(v)$.
Then by $\sigma_1^t(v)\geq 0$ and the inequality above, we have $|\ST_i^t(v)|\leq x_i^t(v) +\sigma_i^t(v)\leq 3x_i^t(v)$.

Hence $\sum_{i\neq 1}|\ST_i^t(v)|\leq 3\sum_{i\neq 1}x_i^t(v)
=3\sum_{i\neq 1}(|\ST_i^t(v)|-\sigma_i^t(v) + (\Delta-1) \sigma_1^t(v))
\leq 3\Delta ( \sigma_1^t(v) + \sum_{i\neq 1}|\ST_i^t(v)|-\sigma_i^t(v))
= 3\Delta x^t(v)$. Here the last equality is because of \Cref{lemma: error for each node alg2}.

Second, consider other cases. By \Cref{lemma: inequalities guarantee in Alg 2}, $\sigma_i^{t}(v)\leq 2\sigma_1^t(v)+1$ or $2\sigma_i^t(v)\leq |\ST_i^t(v)|+1$.

If $\sigma_i^{t}(v)\leq 2\sigma_1^t(v)+1$: $|\ST_i^t(v)|-\sigma_i^t(v)+\sigma_1^t(v)+\nf 12\geq \sigma_1^t(v) + \nf12\geq  \sigma_i^t(v)/2$; If $2\sigma_i^t(v)\leq |\ST_i^t(v)|+1$: $|\ST_i^t(v)|-\sigma_i^t(v)+\sigma_1^t(v)+\nf12\geq |\ST_i^t(v)|-\sigma_i^t(v) + \nf12 \geq  \sigma_i^t(v)$.
Either of them can lead to a conclusion that \[|\ST_i^t(v)|-\sigma_i^t(v)+\sigma_1^t(v)+1/2\geq \sigma_i^t(v)/2.\] Denote $x_i^t(v):=|\ST_i^t(v)|-\sigma_i^t(v)+\sigma_1^t(v)$,
then $|\ST_i^t(v)|\leq x_i^t(v) +\sigma_i^t(v)\leq 3x_i^t(v)+1$.

Consequently $\sum_{i\neq 1}|\ST_i^t(v)|\leq \sum_{i\neq
  i}(3x_i^t(v)+1)= 3\Delta x^t(v)+\Delta $,  where the last equality is because of \Cref{lemma: error for each node alg2}.
\end{proof}

We can finally bound the extra exploration.

\begin{proof}[Proof of~\Cref{lemma: extra exploration totally bounded}]
  Divide the set of nodes on $P^*$ into two sets $A,B$: $A$ contains
  the nodes all of whose children are correctly labeled, and $B$
  contains the other nodes. Then 
  \begin{align}
    \ExEx(t^*) &= \sum_{v\in A} \sum_{i\neq 1}|\ST_i^{t^*}(v)| + \sum_{v\in B} \sum_{i\neq 1}|\ST_i^{t^*}(v)| \\
          &\stackrel{(\star)}{\leq} \sum_{v\in A} 3\Delta x^{t^*}(v) + \sum_{v\in B} (3\Delta x^{t^*}(v)+\Delta ) \\
          &= 3\Delta \sum_{v\in P^*} x^{t^*}(v) + \Delta |B| 
            \stackrel{(\star\star)}{\leq} 6\Delta |\cE^{t^*}| + \Delta |\cE^{t^*}| 
            = 7\Delta |\cE^{t^*}|.
  \end{align}
  The inequality~$(\star)$ uses \Cref{lemma: extra exploration bounded
    for each node alg2}, and $(\star\star)$ uses 
 \Cref{lemma: error for each node alg2}. This proves \Cref{lemma: extra exploration totally bounded}.
\end{proof}

\subsection{Bounding the Movement Cost}
\label{sec:bound-movem-cost}

In this subsection, we bound the total movement cost (and not just the
number of visited nodes), thereby proving~\Cref{lem:distance-by-extra-exploration}.

First, we partition the edge traversals made by the algorithm into
\emph{downwards} (from a parent to a child) and \emph{upwards} (from a
child to its parent) traversals, and denote the cost incurred by the
downwards and upwards traversals until time $t$ by $\cost_d^t$ and
$\cost_u^t$ respectively. We start at the root and hence get
$\cost_d^t = \cost_u^t + d(r,v_t)$; since we care about the time $t^*$
when we reach the goal state $g$, we have
\begin{equation}
  \label{equation: cost to cost-u}
  \cost^{t^*} = \cost_u^{t^*} +\cost_d^{t^*} = 2\cost_u^{t^*} + d(r,v_t).
\end{equation}
It now suffices to bound the upwards movement $\cost_u^{t^*}$. For any edge
$(u,v)$ with $v$ being the parent and $u$ the child, we further partition the
upwards traversals along this edge into two types:
\begin{enumerate}[(i)]
\item upward traversals when the \textbf{if} statement is true at time $t$
  for a node $v_{s}$ (which lies at or below $u$) and we move the
  traversal to another subtree of $\tau(v_s)$ (which lies at or above
  $v$), and % the cost here is
  % the number of already-visited edges we have to traverse, which is
  % the distance from $v_t$ to $\tau(v_t)$.
\item the unique upward traversal when we have completely visited the
  subtree under the edge.
\end{enumerate}

The second type of traversal happens only once, and it never happens
for the edges on the $r$-$g$ path $P^*$ (since those edges contain the
goal state under it, which is not visited until the very end). Hence
the second type of traversals can be charged to the extra exploration
$\ExEx(t^*)$. It remains to now bound the first type of upwards traversals,
which we refer to as \emph{callback} traversals.

We  further partition the callback traversals based on the identity
of the anchor which was critical at that timestep: let
$\cost_u^{t}(v)$ denote the callback traversal cost at those times $s$
when $v = \tau(v_s)$. Hence the total cost of callback traversals is
$\sum_{v\in T^{t^*}} \cost_u^{t^*}(v)$, and
\begin{gather}
  \cost^{t^*} = d(r,v_t) + 2\bigg(\ExEx(t^*) + \sum_{v \in T^{t^*}}
  \cost^{t^*}_u(v) \bigg). \label{eq:1}
\end{gather}
% In
% this subsection, $\cost_u^{t}(v)$ will be bounded by
% $O(\sigma^{t}(v))$, and thus summing up by $v\in T^{t}$ we can get the
% conclusion.
We now control each term of the latter sum.

\begin{lemma}
  \label{lemma: alg2 bounded callback cost for each node}
  For any time $t$ and any node $v\in T^t$, $\cost_u^t(v)\leq 4\sigma^t(v)$.
\end{lemma}
\begin{proof}
  For node $v$ and index $j$, let $S$ be the set of times $s \leq t$ for
  which $v_s \in \ST_j^s(v)$ and
  the \textbf{if} condition is satisfied with $\tau(v_s) = v$ (i.e, $\tau(v_s)=v$, $v$ is active and not degenerate and
      $v$ is
      critical w.r.t. the subtree containing\ $v_s$ at time $s$). The cost of the upwards movement at this
  time is $d(v_{s},v) \leq |\ST_j^{s}(v)|\leq 2\sigma_j^{t_i}(v)$; the
  latter inequality is true by criticality.

%   \Cref{lemma: subtree algorithm lemma alg2} implies that the process
%   restricted to vertices in $T_v$ behaves exactly as though $v$ is the
%   root. 
%   Moreover, the algorithm
\Cref{lemma: not entering if sigma is not minimal}
  ensures that we only enter $\ST_j(v)$
  from a node outside it at some time $s$ when
  $j \in \arg\min_{q}\{\sigma_q^{s}(v)\}$. Hence, if
  $S = \{t_1, \ldots, t_m\}$ then for each $i$ there must exist a time
  $s_i$ satisfying $t_i<s_i<t_{i+1}$ such that
  $\min_{q}\{\sigma_q^{s_i}(v)\}=\sigma_j^{s_i}(v)$. Consequently,
  \[\sigma_j^{t_{i+1}} \geq 2\min_{q}\{\sigma_q^{t_{i+1}}(v)\} \geq
    2\min_{q}\{\sigma_q^{s_i}(v)\}=2\sigma_j^{s_i}(v)\geq
    2\sigma_j^{t_i}(v).\]
  % \alert{AG: The first inequality is strict, by the definition of critical?}
  % \balert{Here the first inequality is by the definition of $t_{i}$, which is the time when $\tau(v_{t_i})$ is critical to $v_{t_i}$. I think it is not strict.}
  Hence, for each $t_i \in S$,
  \begin{gather}
    \sum_{i=1}^m d(v_{t_i},v) \leq \sum_{i=1}^m 2\sigma_j^{t_i}(v)
    \leq 4\sigma_j^{t_m}(v) \leq 4\sigma^t_j(v).
  \end{gather}
  This is the contribution due to a single subtree $T_{\chi_j(v)}$;
  summing over all subtrees gives a bound of $4\sigma^t(v)$, as claimed.
\end{proof}

\begin{proof}[Proof of~\Cref{lem:distance-by-extra-exploration}]
  The equation~(\ref{eq:1}) bounds the total movement cost
  $\cost^{t^*}$ until time $t^*$ in terms of $D$, the extra
  exploration, and the ``callback'' (upwards) traversals
  $\sum_v \cost_u^{t^*}(v)$.  \Cref{lemma: alg2 bounded callback cost for
    each node} above bounds each term $\cost_u^{t^*}(v)$ by
  $4\sigma^{t^*}(v)$. To bound this last summation,
  \begin{itemize}[nolistsep]
  \item For each $v \not\in P^*$, $\sigma^{t^*}(v) = x^{t^*}(v)$ by
    \Cref{lemma: error for each node alg2}. 
  \item For each $v \in P^*$, recall our assumption that
    $g\in \ST_1(v)$, so
    \begin{align*}
    \sum_{v\in P^*} \sigma^{t^*}(v) &= \sum_{v\in P^*} \bigg( \sigma^{t^*}_1(v) + \sum_{i \neq 1}
                      \sigma^{t^*}_i(v) \bigg) \\
                    &\leq \sum_{v \in P^*}  x^{t^*}(v) + \sum_{v \in P^*} \sum_{i \neq 1}
                      |\ST^{t^*}_i(v)| = \sum_{v \in P^*}  x^{t^*}(v) + \ExEx(t^*),
    \end{align*}
    where $\sigma^{t^*}_1(v)\leq x^{t^*}(v)$ is directly given by definition in \Cref{lemma: error for each node alg2}.   
  \end{itemize}
  Summing over all $v$ (using \Cref{lemma: error for each node alg2}), and substituting into~(\ref{eq:1}) gives the claim.
  % Recall that  
  % $$
  % \begin{aligned}
  %   \sum_v \sigma^t(v)
  %   &= \sum_{v\notin P^*}\sigma^t(v) + \sum_{v\in P^*}\sigma^t(v) \\
  %   &= \sum_{v\notin P^*} |X^t(v)| + \sum_{v\in P^*}\sum_i \sigma_i^t(v) && \text{by \cref{lemma: error for each node alg2}} \\
  %   &\leq \sum_{v\notin P^*} |X^t(v)| + \sum_{v\in P^*} |X^t(v)| + \sum_{v\in P^*}\sum_{i\neq 1} \sigma_i^t(v)  && \text{by \cref{lemma: error for each node alg2}} \\
  %   &\leq \sum_{v} |X^t(v)| + \sum_{v\in P^*}\sum_{i\neq 1} |\ST_i^t(v)| \\
  %   &\leq \sum_{v} |X^t(v)| + 4|X^t| && \text{by \Cref{lemma: extra exploration totally bounded}} \\
  %   &\leq 5|X^t|.
  % \end{aligned}
  % $$
\end{proof}

%%% Local Variables:
%%% mode: latex
%%% TeX-master: "main"
%%% End:

\section{The General Tree Exploration Algorithm}
\label{sec: alg for unknown D}

We now build on the ideas from known-distance case to give our
algorithm for the case where the true target distance $d(g,r)$ is not
known in advance, and we have to work merely with the predictions. Recall the guarantee we want to prove:
\UnknownD*

Note that Algorithm \AlgoKnownD requires knowing $D$ exactly in computing \textit{anchors}; an approximation to $D$ does not suffice. Because of this, a simple black-box use of Algorithm \AlgoKnownD using a \textit{``guess-and-double''} strategy %\footnote{\textit{guessing and doubling}: each time guess $\max\{D, \Delta \cdot \cE\}$ to run the algorithm within that budget and then if the algorithm does not succeed, restart the algorithm and double the estimate.} 
does not seem to work.
The main idea behind our algorithm is clean: we explore increasing
portions of the tree. If most of the predictions we see have been
correct, we show how to find a node whose prediction must be
correct. Now running \Cref{alg:degree-k} rooted at this node can solve
the problem. On the other hand, if most of predictions that we have
seen are incorrect, this gives us enough budget to explore
further. 

\subsection{Definitions}
\label{sec: unknown D preliminary}

\begin{definition}[Subtree $\subtree(u,v)$]
  Given a tree $T$, node $v$ and its neighbor $u$, let $\subtree(u,v)$
  denote the set of nodes $w$ such that the path from $w$ to $v$
  contains $u$.
\end{definition}

\begin{lemma}[Tree Separator]
  \label{lemma:lemma-for-alg-vab}
  Given a tree $T$ with maximum degree $\Delta$ and $|T| = n > 2\Delta$
  nodes, there exists a node $v$ and two neighbors $a,b$ such that
  $|\subtree(a,v)|>\frac{|T|}{2\Delta}$ and
  $|\subtree(b,v)|>\frac{|T|}{2\Delta}$. Moreover, such $v,a,b$ can be
  found in linear time.
\end{lemma}

\begin{proof}
  Let $v$ be a \emph{centroid} of tree $T$, i.e., a vertex such that
  deleting $v$ from $T$ breaks it into a forest containing subtrees of
  size at most $n/2$~\cite{Jordan}. Each such subtree corresponds to
  some neighbor of $v$. Let $a,b$ be the neighbors corresponding to
  the two largest subtrees. Then
  $|\subtree(a,v)| \geq \frac{n-1}{\Delta} >
  \frac{n}{2\Delta}$. Moreover the second largest subtree may contain
  $\frac{n-|\subtree(a,v)|-1}{\Delta - 1} \geq \frac{n/2 - 1}{\Delta
    -1 } > \frac{n}{2\Delta}$ when $\Delta < n/2$. 
\end{proof}

\begin{definition}[Vote $\vote(u,c)$ and Dominating vote $\vote(S,c)$]
  Given a center $c$, let the \emph{vote} of any node $u\in T$ be
  $\vote(u,c):=f(u)-d(u,c)$. For any set of nodes $S$, define the
  \emph{dominating vote} to be $\vote(S,c):=x$ if $\vote(u,c) = x$ for
  at least half of the nodes $u\in S$. If such majority value $x$ does
  not exist, define $\vote(S,c) := -1$.
\end{definition}

\subsection{The \TreeX Algorithm}

Given these definitions, we can now give the algorithm. Recall that \Cref{thm:Known-D rigorous} says that \Cref{alg:degree-k} finds $g$ in
  $d(r_\rho,g)  + c_1\Delta\cdot |\cE|$ steps, for some constant $c_1\geq 1$. We proceed in
rounds: in round $\rho$ we run \Cref{alg:degree-k} and visit
approximately $\Delta \cdot (c_1+\beta)^{\rho}$ vertices, where
$\beta \geq 1$ is a parameter to be chosen later. % (If this number times $\beta$ is larger than the estimated depth, run more points to see whether the goal can be found.
Now we focus on two
disjoint and ``centrally located'' subtrees of size $\approx (c_1+\beta)^{\rho}$
within the visited nodes. Either the majority of these nodes have
correct predictions, in which case we use their information to
identify one correct node. Else a majority of them are incorrect, in
which case we have enough budget to go on to the next round. A formal
description appears in \Cref{alg: alg for Graph Searching D not known}.
% \alert{Changed $t$ to $\rho$ here: please check.}

\begin{algorithm}[h]
  \caption{$\TreeX(r, \beta)$}
  \label{alg: alg for Graph Searching D not known}
  $r_0\gets r$, $D_0\gets f(v)$, $\rho\gets 0$\;
  \While{goal $g$ not found}{
    $B_{\rho} \gets (c_1+\beta)^{\rho} \cdot (2\Delta+1)$ \;
    \If{$B_\rho< D_\rho/\beta$}{
        run $\AlgoKnownD(r_{\rho},D_{\rho}, B_{\rho})$ \;}
    \Else{
        run $\AlgoKnownD(r_{\rho},D_{\rho}, D_\rho + c_1B_{\rho})$ \;
    }
    $T^{\rho+1}\gets$ tree induced by nodes that have ever been visited
    so far \;
    $r_{\rho+1},a_{\rho+1},b_{\rho+1}\gets$ centroid for $T^{\rho}$ and its two neighbors promised by~\Cref{lemma:lemma-for-alg-vab}\;
    let $D_{a,\rho+1}\gets \vote(\subtree(a_{\rho+1}, r_{\rho+1}), r_{\rho+1})$ and $D_{b,\rho+1}\gets \vote(\subtree(b_{\rho+1}, r_{\rho+1}), r_{\rho+1})$\;
    define new distance estimate $D_{\rho+1}\gets \max\{D_{a,\rho+1},D_{b,\rho+1}\} $\;
    move to vertex $r_{\rho+1}$\;
    $\rho\gets \rho+1$\;
  }
\end{algorithm}

\subsection{Analysis of the \TreeX Algorithm}

\begin{lemma}
  \label{lemma: D not known find correct D'}
  If the goal is not
  visited before round $\rho$ when $B_\rho \geq 4|\cE|(2\Delta+1)$, we have $D_\rho = d(r_\rho,g)$.
\end{lemma}
\begin{proof}
First, if $|\cE|=0$, then the conclusion holds obviously. So next we assume $|\cE|>0$.
  The execution of \Cref{alg:degree-k} in round $\rho-1$ visits at least
  $B_{\rho-1} = (c_1+\beta)^{(\rho-1)}\cdot (2\Delta+1)$ distinct nodes. Using
  the assumption on $B_\rho$, we have
  \[ |T^\rho|\geq 4|\cE|\cdot (2\Delta+1) > 4\Delta|\cE| >
    2\Delta. \] \Cref{lemma:lemma-for-alg-vab} now implies that both
  the subtrees $\subtree(a_{\rho},r_{\rho})$ and
  $\subtree(b_\rho,r_\rho)$ contain more than
  $\frac{1}{2\Delta}|T^\rho|>2|\cE|$ nodes. Since at most $|\cE|$ nodes
  are erroneous, more than half of the nodes in each of 
  $\subtree(a_\rho,r_\rho)$ and $\subtree(b_\rho,r_\rho)$ have correct
  predictions.

  Finally, observe that if $g \not\in \subtree(a_\rho,r_\rho)$, then
  for any correct node $x$ in $\subtree(a_\rho,r_\rho)$ we have
  $f(x) = d(x,g) = d(x,r_\rho)+d(r_\rho,g)$, and hence its vote
  $\vote(x,r_\rho) = d(r_\rho,g)$. Since a majority of nodes in
  $\subtree(a_\rho,r_\rho)$ are correct, we get
  \begin{gather}
    D_{a,\rho} = \vote(\subtree(a_{\rho}, r_{\rho}), r_{\rho}) =
    d(r_\rho,g). \label{eq:2}
  \end{gather}
   On the other hand, if
  $g \in \subtree(a_\rho,r_\rho)$, then for any correct node $x$ in
  $\subtree(a_\rho,r_\rho)$ we have
  $f(x) = d(x,g) \leq d(x,a_\rho)+d(a_\rho,g) <
  d(x,r_\rho)+d(r_\rho,g) $. Thus its vote, and hence the vote of a
  strict majority of nodes in the subtree $\subtree(a_\rho,r_\rho)$
  have
  \begin{gather}
    D_{a,\rho} < d(r_\rho,g). \label{eq:3}
  \end{gather}
  If no value is in a strict majority, recall that we define
  $D_{a,\rho} = -1$, which also satisfies~(\ref{eq:3}). The same
  arguments hold for the subtree $\subtree(b_\rho,r_\rho)$ as
  well. Since the goal $g$ belongs to at most one of these subtrees,
  we have that $D_\rho = \max(D_{a,\rho}, D_{b,\rho}) = d(r_\rho,g)$,
  as claimed.
\end{proof}

\begin{lemma}
\label{lemma: D not known vt distance new}
For any round $\rho$, $d(r_\rho,r)\leq O(B_{\rho})$.
Moreover, for any round $\rho$ such that $B_\rho\geq 4|\cE|(2\Delta+1)$, $d(r_\rho,r)\leq O(B_{\rho-1})+ O(\beta|\cE|\Delta)$.
\end{lemma}

\begin{proof}
Since $r_{\rho}$ is at distance at most $(c_1+c_3)B_{\rho-1}=B_\rho$ from
  $r_{\rho-1}$, an inductive argument shows that its distance from
  $r_0 = r$ is at most $(B_0+ \cdots + B_{\rho}) =O(B_{\rho})$. 
  
  Moreover, when $B_\rho\geq 4|\cE|(2\Delta+1)$, we have $d(r_\rho,g)=D_\rho$ by \Cref{lemma: D not known find correct D'}. Hence if $B_\rho \geq D_\rho/\beta$, the algorithm finds the goal in this round by \Cref{thm:Known-D rigorous}. 
Therefore, for any rounds $\rho$ when $B_\rho\geq 4|\cE|(2\Delta+1)$ except the last round, the number of nodes visited by \Cref{alg:degree-k} is at most $B_\rho$, hence we have $d(r_{\rho+1},r)\leq d(r_{\rho},r)+B_\rho$. We denote $\rho'$ to be the first round $\rho'$ such that $B_{\rho'}\geq 4|\cE|(2\Delta+1)$. Thus by induction we have \[
d(r_{\rho},r) \leq \sum_{i=\rho'}^{\rho-1} B_{i} + d(r_{\rho'},r) \leq
O(B_{\rho-1})+O(B_{\rho'})\leq O(B_{\rho-1})+
O(\beta|\cE|\Delta). \qedhere \] 
\end{proof}

\begin{proof}[Proof of~\Cref{thm:main}]
  Firstly, for the rounds $\rho$ when $B_\rho<4|\cE|(2\Delta+1)$: in
  each round, \Cref{alg:degree-k} at most visits
  $(c_1+\beta)B_\rho=B_{\rho+1}$ nodes, the cost incurred is at most
  $19B_{\rho+1}$,
  by~\Cref{lem:distance-by-extra-exploration}. Moreover, the distance
  from the ending node to $r_{\rho+1}$ is a further $O(B_{\rho+1})$ by
  \Cref{lemma: D not known vt distance new}. Therefore, since the
  bounds $B_\rho$ increase geometrically, the cost summed over all
  rounds until round $\rho$ is $O(B_{\rho+1})=O(\beta|\cE|\Delta)$.
  
  Secondly, for any rounds $\rho$ when $B_\rho\geq 4|\cE|(2\Delta+1)$
  except the last round, by \Cref{lemma: D not known find correct D'}
  and \Cref{thm:Known-D rigorous}, the number of nodes visited by
  \Cref{alg:degree-k} is at most $B_\rho$ (the reasoning is the same
  as that in \Cref{lemma: D not known vt distance new}). Hence the
  cost incurred is at most $19B_\rho$. Moreover, by \Cref{lemma: D not
    known vt distance new} the distance from the ending node to
  $r_{\rho +1}$ is at most $O(B_\rho) + O(\beta\Delta |\cE|)$, which
  means the total cost in round $\rho$ is at most
  $O(B_\rho)+ O(\beta\Delta |\cE|)$.

  Moreover, if we denote round $\rho'$ to be the first round such that
  $B_{\rho'} \geq 4|\cE|(2\Delta+1)$, then we have, for any round
  $\rho>\rho'$, $B_\rho>\beta \Delta |\cE|$. Hence the cost in round
  $\rho$ is $O(B_\rho)$.

  Finally, consider the last round $\rho^*$. We only need to consider the
  case when $B_{\rho^*} \geq 4|\cE|(2\Delta+1)$, otherwise the cost
  has been included in the first case. By \Cref{thm:Known-D rigorous},
  the cost incurred in this round is at most
  $D_{\rho^*}+c_1\Delta|\cE|\leq d(r,g)+d(r_{\rho^*},r) + c_1\Delta
  |\cE|$.  So summing the bounds above, the total cost is at most
  \begin{align*}
    &O(\beta\Delta |\cE|) + O(B_{\rho'})+O(\beta\Delta |\cE|) +
      \sum_{i=\rho'+1}^{\rho^*-1}O(B_{i}) +
      d(r,g)+d(r_{\rho^*},r)+c_1\Delta |\cE| \\ &\qquad \leq d(r,g)+ O(B_{\rho^*-1}) +O(\beta\Delta |\cE|)
                                                  \leq d(r,g) + O(d(r,g)/\beta)+O(\beta\Delta |\cE|) 
  \end{align*}
  Here the final inequality uses that
  \[B_{\rho^*-1}\leq D_{\rho^*-1}/\beta\leq (d(r,g)+O(\beta
  B_{\rho^*-1}))/\beta\leq (d(r,g)+O(B_{\rho^*-1}))/\beta.\] 
  Setting $\beta = O(1/\delta)$ gives the proof.
\end{proof}

%%% Local Variables:
%%% mode: latex
%%% TeX-master: "main"
%%% End:

\section{The Planning Problem}
\label{sec:full-information}

In this section we consider the planning version of the problem when
the entire graph $G$ (with unit edge lengths, except for \S\ref{sec:analys-bound-doubl-2}), the starting node $r$, and the entire prediction
function $f:V\to \mathbb{Z}$ are given up-front. The agent can use
this information to plan its exploration of the graph. We propose an
algorithm for this version and then prove the cost bound for trees,
and then for a graph with bounded doubling dimension. We begin by defining the \emph{implied-error} function $\varphi(v)$, which
gives the total error if the goal is at node $v$.

\begin{definition}[Implied-error]
  The \emph{implied-error function} $\varphi: V \to \ZZ$ maps each
  node $v\in V$ to $\varphi(v):=|\{u \in V \mid d(u,v)\neq f(u)\}|$,
  which is the $\ell_0$ error if the goal were at $v$.
\end{definition}

The search algorithm for this planning version is particularly simple:
we visit the nodes in rounds, where round $\rho$ visits nodes with
implied-error $\varphi$ value at most $\approx 2^\rho$ in the
cheapest possible way.
% \alert{The index $t$ in this algorithm refers to rounds and
%   not timesteps; change to $\rho$? Not a big deal.}
The challenge is to show that the total cost incurred until reaching
the goal is small. Observe that $|\cE| = \varphi(g)$, so if this value
is at most $2^\rho$, we terminate in round $\rho$. 

\begin{algorithm}[h]
  \caption{\FullInfoX}
  \label{alg: search according to phi}
  $\rho\gets 0$, $S_{-1} \gets \emptyset$, $r_{-1} \gets r$\;
  \While{$g$ not found}{
    $S_\rho\gets \{v\in T \mid \varphi(v) < 2^\rho\} \setminus (\cup_{i =
      -1}^{\rho-1} S_i)$\;
    \eIf{$S_\rho \neq \emptyset$}{
      $C_\rho\gets$ min-length Steiner Tree on $S_\rho$\;
      go to an arbitrary node $r_\rho$ in $S_\rho$\;
      visit all nodes in $C_\rho$ using an Euler tour of cost at most
      $2|C_\rho|$, and return to $r_\rho$\;
    }{$r_\rho \gets r_{\rho-1}$}
    $\rho \gets \rho+1$\;
  }
\end{algorithm}

\subsection{Analysis}
\label{sec:analysis-full-info}

Recall our main claim for the planning algorithm:
\FullInfo*

The proof relies on the fact that
\Cref{alg: search according to phi} visits a node in $S_\rho$ only after
visiting all nodes in $\cup_{s < \rho} S_s$ and not finding the goal $g$;
this serves a proof that $|\cE| = \varphi(g) \geq 2^\rho$. 
The proof below shows that (a)~the cost of the tour of $C_\rho$ is
bounded and (b)~the total cost of each transition is small. 
% The proof
% below shows that (a)~the cost of the tour of $C_t$ is $O(\Delta \, 2^t)$,
% and moreover (b)~the total cost of each transition is small.
Putting these claims together then proves \Cref{thm:full-info}. We start with a definition.

\begin{definition}[Midpoint Set]
  Given a set of nodes $U$, define its \emph{midpoint set} $M(U)$ to be the set
  of points $w$ such that the distance from $w$ to all points in $U$
  is equal. 
\end{definition}

\begin{lemma}[$\varphi$-Bound Lemma]
  \label{lemma: phi and Mid set}
  For any two sets of nodes $S, U \subseteq G$, we have
  \[ \sum_{v\in U} \varphi(v) \geq |S\setminus M(U)|.\]
\end{lemma}
\begin{proof}
  If node $w \in S$ does not lie in $M(U)$, then there are two nodes
  $u,v\in U$ for which $d(u,w)\neq d(v,w)$. This means $f(w)$ cannot
  equal both of them, and hence contributes to at least one of
  $\varphi(u)$ or $\varphi(v)$.
\end{proof}

\begin{corollary}
  \label{cor: distance less than phi}
  For any two nodes $u,v\in G$, we have $d(u,v)\leq \varphi(u)+\varphi(v)$.
\end{corollary}

\begin{proof}
  Apply \Cref{lemma: phi and Mid set} for set $U = \{u,v\}$ and $S$
  being a (shortest) path between them (which includes both
  $u,v$). All edges have unit lengths so $|S| = d(u,v)+1$;
  moreover, $|M(U) \cap S| \leq 1$.
\end{proof}

\subsubsection{Analysis for Trees (\Cref{thm:full-info}(i))}
\label{sec:analys-trees-crefthm}

\begin{lemma}[Small Steiner Tree]
  \label{lemma: full-info C bound}
  If $\rho = 0$ then $|C_\rho| = 1$ else $|C_\rho| \leq O( \Delta \cdot 2^\rho)$.
\end{lemma}
\begin{proof}
  If $\rho=0$, then $S_\rho$ contains all nodes with $\varphi(v) = 0$; there
  can be only one such node. Else if $|S_\rho|\leq 1$ then
  $|C_\rho|\leq 1 \leq 2^\rho$, so assume that $|S_\rho|>1$ and let
  $u_1,u_2 := \arg\max_{u,v\in S_\rho}\{d(u,v)\}$ be a farthest pair of
  nodes in $S_\rho$. Consider path $p$ from $u_1$ to $u_2$: if all nodes
  $w\in p$ have $d(w,u_1) \neq d(w,u_2)$, then the midpoint set
  $|M(\{u_1,u_2\})|=0$, so \Cref{lemma: phi and Mid set} says
  $|C_\rho|\leq \varphi(u_1)+\varphi(u_2) \leq 2\times 2^{\rho} = 2^{\rho+1}$,
  giving the proof. Hence, let's consider the case where there exists
  $w\in p$ with $d(w,u_1)=d(w,u_2)$.

  Let $w$'s neighbors in $C_\rho$ be $q_1,\ldots,q_k$ for some
  $k\leq \Delta$. If we delete $w$ and its incident edges, let $C_{\rho,i}$
  be the subtree of $C_\rho$ containing $q_i$; suppose that
  $u_1 \in C_{\rho,1}$ and $u_2 \in C_{\rho,2}$. Choose any arbitrary vertex
  $u_i \in (C_{\rho,i} \cap S_\rho)$; such a vertex exists because $C_\rho$ is a
  min-length Steiner tree connecting $S_\rho$. Let
  $U := \{u_1, \ldots, u_k\}$.

  Consider any node $x\neq w$ in $C_\rho$: this means $x \in C_{\rho,j}$ for
  some $j$. Choose $i \in \{1,2\}$ such that $i \neq j$. By the tree
  properties, $d(x,u_i)=d(x,w)+d(w,u_i)$. Moreover, we have
  $d(u_i,u_{2-i})\geq d(u_j, u_{2-i})$ by our choice of $\{u_1, u_2\}$,
  so $d(w,u_i)\geq d(w,u_j)$.  This
  means
  \[d(x,u_i)=d(x,w)+d(w,u_i)\geq
    d(x,w)+d(w,u_j)=d(x,q_j)+d(u_j,q_j)+2> d(x,u_j),\] which means
  $x\notin M(U)$. In summary, $M(U)=\{w\}$ or $|M(U)|=0$, so applying
  \Cref{lemma: phi and Mid set} in either case gives \[ |C_\rho|\leq
  |C_\rho\setminus M(U)|+1\leq \sum_{i=1}^k \varphi(u_i)+1\leq
  \Delta\cdot (2^{\rho}+1). \qedhere\]
\end{proof}
% \begin{wrapfigure}{R}{0.3\textwidth}
%   \centering
%    \input{./Figures/Full-info-Ball.tex}
%   \caption{Let $u^*,v^*$ be the diameter of the set $S_\rho$ (i.e, $u^*,v^*=\text{argmax}_{u,v\in S_\rho}d(u,v)$). $c$ is any node in $N$ and $B(c)$ is its neighbor. We show in \Cref{clm:ball-distance} that the size of $B(c)$ is $O(2^\rho)$.}
% \vspace{-0.5cm}
% \label{fig: Full-info ball}
% \end{wrapfigure}

\begin{lemma}[Small Cost for Transitions]
  \label{lemma: full-info cost at iteration i}
  Consider the first round $\rho_0$ such that $r_{\rho_0} \neq r$, then
  $d(r,r_{\rho_0}) \leq d(r, g) + |\cE| + 2^{\rho_0} \mathbf{1}_{({\rho_0} > 0)}$. For each subsequent round $\rho>\rho_0$,
  $d(r_{\rho-1}, r_\rho) \leq 2^{\rho+1}$.
\end{lemma}

\begin{proof}
  If the first transition happens in round ${\rho_0}$, its cost is
  \[ d(r, r_{\rho_0}) \leq d(r,g) + d(g,r_{\rho_0}) \leq d(r,g) + \varphi(g) +
    \varphi(r_{\rho_0}) \leq d(r,g) + |\cE| + 2^{\rho_0} \mathbf{1}_{({\rho_0} > 0)}, \] where we
  used~\Cref{cor: distance less than phi} for the second inequality.
  For all other transitions, \Cref{cor: distance less than phi} again
  gives $d(r_{\rho-1}, r_\rho) \leq \varphi(r_{\rho-1}) + \varphi(r_\rho) \leq
  2^{\rho-1} + 2^\rho \leq 2^{\rho+1}$.
\end{proof}

\begin{proof}[Proof of~\Cref{thm:full-info}(i)]
  Suppose $g$ belongs to $S_\rho$, then
  $|\cE| \geq 2^{\rho-1} \cdot \mathbf{1}_{\rho > 0}$. But now the cost over
  all the transitions is at most
  $d(r,g) + |\cE| + O(2^\rho) \cdot \mathbf{1}_{\rho > 0}$ by summing the
  results of \Cref{lemma: full-info cost at iteration i}. The cost of
  the Euler tours are at most $\sum_{s \leq \rho} 2(|C_s|-1)$ by \Cref{lemma: full-info C bound}, which
  gives at most $O(\Delta \cdot 2^\rho) \cdot \mathbf{1}_{\rho >
    0}$. Combining these proves the theorem.
\end{proof}

\subsection{Analysis for Bounded Doubling Dimension (\Cref{thm:full-info}(ii))}
\label{sec:analys-bound-doubl}

For a graph $G = (V,E)$ with doubling dimension $\alpha$, and
unit-length edges, we consider running \Cref{alg: search according to
  phi}, as for the tree case.
% \begin{theorem}
%   \label{thm:full-info-doubling-unit}
%   For graph exploration on arbitrary graphs with unit edge lengths,
%   \Cref{alg: search according to phi} incurs a cost
%   $d(r,g) + 2^{O(\alpha)} \cdot O(|\cE|^2)$.
% \end{theorem}
We merely replace \Cref{lemma: full-info C bound} by the following
lemma, and the rest of the proof is the same as the proof of the tree case:
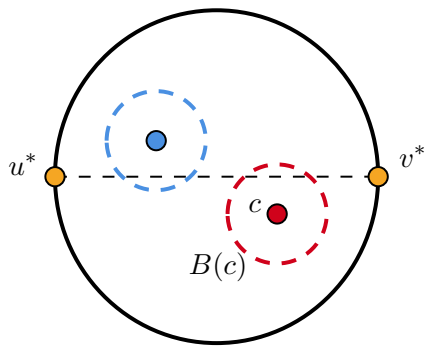
\begin{wrapfigure}{R}{0.4\textwidth}
\vspace{-0.8cm}
  \centering
   \tikzset{every picture/.style={line width=0.75pt}} %set default line width to 0.75pt        

\begin{tikzpicture}[x=0.75pt,y=0.75pt,yscale=-1,xscale=1]

%Shape: Ellipse [id:dp5455760784369359] 
\draw  [line width=1.5]  (110.95,152.2) .. controls (110.94,105.81) and (147.43,68.19) .. (192.44,68.18) .. controls (237.45,68.18) and (273.94,105.78) .. (273.95,152.17) .. controls (273.96,198.56) and (237.48,236.18) .. (192.47,236.18) .. controls (147.46,236.19) and (110.96,198.59) .. (110.95,152.2) -- cycle ;
%Straight Lines [id:da39919210620907486] 
\draw  [dash pattern={on 4.5pt off 4.5pt}]  (110.95,152.2) -- (273.95,152.17) ;
%Shape: Circle [id:dp86972200079447] 
\draw  [fill={rgb, 255:red, 245; green, 166; blue, 35 }  ,fill opacity=1 ] (115.78,152.71) .. controls (116.06,150.04) and (114.13,147.65) .. (111.46,147.37) .. controls (108.79,147.09) and (106.41,149.03) .. (106.13,151.69) .. controls (105.85,154.36) and (107.78,156.74) .. (110.45,157.02) .. controls (113.11,157.3) and (115.5,155.37) .. (115.78,152.71) -- cycle ;
%Shape: Circle [id:dp3080307713475675] 
\draw  [fill={rgb, 255:red, 245; green, 166; blue, 35 }  ,fill opacity=1 ] (278.8,152.39) .. controls (278.92,149.71) and (276.85,147.44) .. (274.17,147.32) .. controls (271.49,147.2) and (269.23,149.28) .. (269.11,151.95) .. controls (268.99,154.63) and (271.06,156.9) .. (273.74,157.02) .. controls (276.41,157.14) and (278.68,155.06) .. (278.8,152.39) -- cycle ;
%Shape: Circle [id:dp4756821139341677] 
\draw  [color={rgb, 255:red, 74; green, 144; blue, 226 }  ,draw opacity=1 ][fill={rgb, 255:red, 0; green, 0; blue, 0 }  ,fill opacity=0 ][dash pattern={on 5.63pt off 4.5pt}][line width=1.5]  (137,134) .. controls (137,120.19) and (148.19,109) .. (162,109) .. controls (175.81,109) and (187,120.19) .. (187,134) .. controls (187,147.81) and (175.81,159) .. (162,159) .. controls (148.19,159) and (137,147.81) .. (137,134) -- cycle ;
%Shape: Circle [id:dp6645986648807133] 
\draw  [color={rgb, 255:red, 208; green, 2; blue, 27 }  ,draw opacity=1 ][dash pattern={on 5.63pt off 4.5pt}][line width=1.5]  (198,171) .. controls (198,157.19) and (209.19,146) .. (223,146) .. controls (236.81,146) and (248,157.19) .. (248,171) .. controls (248,184.81) and (236.81,196) .. (223,196) .. controls (209.19,196) and (198,184.81) .. (198,171) -- cycle ;
%Shape: Circle [id:dp14759736395701928] 
\draw  [fill={rgb, 255:red, 74; green, 144; blue, 226 }  ,fill opacity=1 ] (166.85,134.22) .. controls (166.97,131.54) and (164.89,129.27) .. (162.22,129.15) .. controls (159.54,129.03) and (157.27,131.11) .. (157.15,133.78) .. controls (157.03,136.46) and (159.11,138.73) .. (161.78,138.85) .. controls (164.46,138.97) and (166.73,136.89) .. (166.85,134.22) -- cycle ;
%Shape: Circle [id:dp17091364475394943] 
\draw  [fill={rgb, 255:red, 208; green, 2; blue, 27 }  ,fill opacity=1 ] (227.85,171.22) .. controls (227.97,168.54) and (225.89,166.27) .. (223.22,166.15) .. controls (220.54,166.03) and (218.27,168.11) .. (218.15,170.78) .. controls (218.03,173.46) and (220.11,175.73) .. (222.78,175.85) .. controls (225.46,175.97) and (227.73,173.89) .. (227.85,171.22) -- cycle ;

% Text Node
\draw (86,139) node [anchor=north west][inner sep=0.75pt]   [align=left] {$\displaystyle u^{*}$};
% Text Node
\draw (283,135) node [anchor=north west][inner sep=0.75pt]   [align=left] {$\displaystyle v^{*}$};
% Text Node
\draw (207,162) node [anchor=north west][inner sep=0.75pt]   [align=left] {$\displaystyle c$};
% Text Node
\draw (177,189) node [anchor=north west][inner sep=0.75pt]   [align=left] {$\displaystyle B( c)$};

\end{tikzpicture}
  \caption{Let $u^*,v^*$ be the diameter of set $S_\rho$ (i.e, $u^*,v^*=\text{argmax}_{u,v\in S_\rho}d(u,v)$). $c$ is any node in $N$ and $B(c)$ is its neighbor. We show in \Cref{clm:ball-distance} that the size of $B(c)$ is $O(2^\rho)$.}
\vspace{-0.5cm}
\label{fig: Full-info ball}
\end{wrapfigure}
% \begin{figure}[ht]
% \centering
%     \input{./Figures/Full-info-Ball.tex}
%     \caption{Let $u^*,v^*$ be the diameter of the set $S_\rho$ (i.e, $u^*,v^*=\text{argmax}_{u,v\in S_\rho}d(u,v)$). $c$ is any node in $N$ and $B(c)$ is its neighbor. We show in \Cref{clm:ball-distance} that the size of $B(c)$ is $O(2^\rho)$.}
%     \label{fig: Full-info ball}
% \end{figure}

\begin{lemma}
  \label{lem:full-info-C-bound-doubling}
%   If $\rho = 0$ then $|C_\rho| = 1$ else t
  The total length of the tree
  $C_\rho$ is at most $2^{O(\alpha)}\cdot 2^{2\rho}$.   
\end{lemma}

\begin{proof}
%   \alert{Handle the $\rho=0$ case. Copied this over from the general
%     weights case. See if we can avoid duplication.} 
  If $|S_\rho|\leq 1$, then $|C_\rho|\leq 1$. Hence next we assume that $|S_\rho|\geq 2$. 
  Define $R := \max_{u,v \in S_\rho} d(u,v)$, and let
  $u^*, v^* \in S_\rho$ be some points at mutual distance $R$.  Let
    $N$ be an $R/8$-net of $S_\rho$. (An $\e$-net $N$ for a set $S$ satisfies the properties (a)~$d(x,y) \geq \e$ for all $x,y \in N$, and (b) for all $s \in S$ there exists $x \in N$ such that $d(x,s) \leq \e$.) Since the metric has doubling
  dimension $\alpha$, it follows that
    $|N| \leq (\frac{R}{R/8})^{O(\alpha)} = 2^{O(\alpha)}$~\cite{GuptaKL03}.  Let each
  point in $S_\rho$ choose a closest net point (breaking ties
  arbitrarily), and let $B(c) \sse S_\rho$ be the points that chose 
  $c \in N$ as their closest net point (see \Cref{fig: Full-info ball} for a sketch).

  \begin{claim}
    \label{clm:ball-distance}
    For each net point $c\in N$, we have $|B(c)| \leq O(2^\rho)$.
  \end{claim}

  \begin{proof}
  Because $d(v^*,c)+d(u^*,c)\geq d(u^*,v^*) = R$, hence without loss of generality we assume $d(v^*,c)\geq R/2$. For any point $w\in B(c)$, $d(w,v^*)\geq d(v^*,c)-d(c,w)\geq R/2-R/8 > R/8\geq d(w,c)$. Hence $w$ is not in $M(\{c,v^*\})$. Hence by Lemma~\ref{lemma: phi and Mid set}, \[2^{\rho+1}\geq \varphi(c)+\varphi(v^*)\geq |S_\rho \setminus M(\{v^*,c\})| \geq |B(c)|. \qedhere \]
  \end{proof}

  There are $2^{O(\alpha)}$ net points, so
  $|S_\rho| \leq 2^{O(\alpha)} \cdot 2^\rho$. Finally, \Cref{cor:
    distance less than phi} holds for general unit-edge-length
  graphs, so the cost of connecting any two nodes in $S_\rho$ is at
  most $2^\rho$, and therefore $|C_\rho| \leq 2^{O(\alpha)} \cdot 2^{2\rho}$.
\end{proof}

Using \Cref{lem:full-info-C-bound-doubling} instead of \Cref{lemma:
  full-info C bound} in the proof of \Cref{thm:full-info}(i) gives the
claimed bound of $2^{O(\alpha)}\cdot |\cE|^2$, and completes the proof
of \Cref{thm:full-info}(ii).

\subsection{Analysis for Bounded Doubling Dimension: Integer Lengths}
\label{sec:analys-bound-doubl-2}

In this part, we further generalize the proof above to the case when
the edges can have positive integer lengths.
% Also, we now need to
% consider the $\ell_1$-error instead of the $\ell_0$-error.
Consider an graph $G = (V,E)$ with doubling dimension $\alpha$ and
general (positive integer) edge lengths. Define the $\ell_1$ analog of
the implied-error function to be:
\[ \varphi_1(v) := \sum_{u \in V} |f(u) - d(u,v)|. \] Since we are in
the full-information case, we can compute the $\varphi_1$ value for
each node. Observe that $\varphi_1(g)$ is the $\ell_1$-error; we prove
the following guarantee.

\begin{theorem}
  \label{thm:full-info-doubling-lengths}
  For graph exploration on arbitrary graphs with positive integer edge
  lengths, the analog of \Cref{alg: search according to phi} that uses
  $\varphi_1$ instead of $\varphi$, incurs a cost
  $d(r,g) + 2^{O(\alpha)} \cdot O( \varphi_1(g))$.
\end{theorem}

The proof is almost the same as that for the unit length case. We merely replace~\Cref{cor: distance less than phi} and \Cref{clm:ball-distance} by the following two lemmas.

\begin{lemma}
  \label{lem:weighted-distance-phi}
  For any two vertices $u,v$, their distance $d(u,v) \leq \nf12( \varphi_1(u) + \varphi_1(v))$.
\end{lemma}
\begin{proof}
By definition of $\varphi_1$ we have
  $\varphi_1(u)+\varphi_1(v)\geq |f(u)|+|f(v)-d(u,v)| + |f(u)-d(u,v)|+|f(v)|\geq 2d(u,v)$.
\end{proof}

\begin{claim}
    \label{clm:ball-distance int length}
    For each net point $c\in N$, we have
    $\sum_{v \in B(c)} d(v,u^*) \leq O(2^\rho)$.
  \end{claim}

  \begin{proof}
    Let $w$ be the node among $u^*, v^*$ that is further from $c$; by
    the triangle inequality, $d(c,w) \geq R/2$. By the properties of
    the net, $d(v,c) \leq R/8$. Again using the triangle inequality,
    $d(v,w) \geq 3R/8$. Hence
    \[ \varphi_1(w) + \varphi_1(c) \geq \sum_{v \in B(c)} \big( |f(v) -
      d(v,w)| + |f(v) - d(v,c)| \big) \geq |B(c)| \cdot (\nf{3R}8 -
      \nf{R}8).\] Since both $w,c \in S_\rho$, this implies that
    \[ |B(c)|\cdot R \leq 4(\varphi_1(w) + \varphi_1(c)) \leq
      O(2^\rho). \] Finally, we use that $d(v,u^*) \leq R$ by our choice
    of $R$ to complete the proof.
  \end{proof}

Now to prove
\Cref{thm:full-info-doubling-lengths}, we mimic the
proof of~\Cref{thm:full-info}(ii), just substituting 
\Cref{lem:weighted-distance-phi} and
\Cref{clm:ball-distance int length} instead of \Cref{cor:
  distance less than phi} and \Cref{clm:ball-distance}.

%%% Local Variables:
%%% mode: latex
%%% TeX-master: "main"
%%% End:

%\newpage
%\input{questions}
\section{Closing Remarks}

In this paper we study a framework for graph exploration problems with
predictions: as the graph is explored, each newly observed node gives
a prediction of its distance to the goal. While graph searching is a
well-explored area, and previous works have also studied models where
nodes give directional/gradient information (``which neighbors are
better''), such distance-based predictions have not been previously
studied, to the best of our knowledge. We give algorithms for
exploration on trees, where the total distance traveled by the agent
has a relatively benign dependence on the number of erroneous
nodes. We then show results for the planning version of the problem,
which gives us hope that our exploration results may be extendible to
broader families of graphs. This is the first, and most natural open
direction. 

Another intriguing direction is to reduce the space complexity of our
algorithms, which would allow us to use them on very large implicitly
defined graphs (say computation graphs for large dynamic programming
problems, say those arising from reinforcement learning problems, or
from branch-and-bound computation trees). Can we give time-space
tradeoffs? Can we extend our results to multiple agents? 
% \todo{More
%   questions and discussion?  Less?}
  A more open-ended direction is to
consider other forms of quantitative hints for graph searching, beyond
distance estimates (studied in this paper) and gradient information (studied in previous works).

% \subsection*{Acknowledgments}

% We thank Christian Coester for suggesting 

{\small
\bibliographystyle{alpha}
\bibliography{sample,bibliography}
}

%\section{Other Results}

% \todo{Write something here.}

% \subsection{Other Notions of Error}

% \begin{enumerate}
% \item $\ell_0$ vs $\ell_1$ error? What if predictions are fractional
%   but distances are integral? Other reasonable notions? 
% \end{enumerate}

\section{Further Discussion}

\subsection{$\ell_0$-versus-$\ell_1$ Error in Suggestions}
\label{sec:ell0-versus-ell1}

Most of the paper deals with $\ell_0$ error: namely, we relate our
costs to $|\cE|$, the number of vertices that give incorrect
predictions of their distance to the goal. Another reasonable notion
of error is the $\ell_1$ error: $\sum_v |f(v) - d(v,g)|$.

For the case of integer edge-lengths and integer predictions, both of
which we assume in this paper, it is immediate that the $\ell_0$-error
is at most the $\ell_1$-error: if $v$ is erroneous then the former
counts $1$ and the latter at least $1$. If we are given integer
edge-lengths but fractional predictions, we can round the predictions
to the closest integer to get integer-valued predictions $f'$, and
then run our algorithms on $f'$. Any prediction that is incorrect in
$f'$ must have incurred an $\ell_1$-error of at least $\nf12$ in $f$.
Hence all our results parameterized by the $\ell_0$ error imply
results parameterized with the $\ell_1$ error as well.

\subsection{Extending to General Edge-Lengths}
\label{sec:general-lengths}

A natural question is whether a guarantee like the one proved in
\Cref{thm:main} can be shown for trees with general integer weights:
let us see why such a result is not possible.
\begin{enumerate}
\item The first observation is that the notion of error needs to be
  changed from $\ell_0$ error something that is homogeneous in the
  distances, so that scaling distances by $C>0$ would change the error
  term by $C$ as well. One such goal is to guarantee
  the total movement to be 
  \[ O(d(r,g) + \text{ some function of the $\ell_p$ error} ), \]
  where $\ell_p$-error is $(\sum_v |f(v) - d(v,g)|^p)^{1/p}$. 
\item Consider a complete binary tree of height $h$, having $2^h$
  leaves. Let all edges between internal nodes have length $0$, and
  edges incident to leaves have length $L \gg 1$. The goal is at one
  of the leaves. Let all internal nodes have $f(v) = L$, and let all
  leaves have prediction $2L$. Hence the total $\ell_p$ error is $2L$,
  whereas any algorithm would have to explore half the leaves in
  expectation to find the goal; this would cost $\Theta(2^h \cdot L)$,
  which is unbounded as $h$ gets large.
\item The problem is that zero-length edges allow us to simulate
  arbitrarily large degrees. Moreover, the same argument can be
  simulated by changing zero-length edges to unit-length edges; the
  essential idea remains the same. and setting $f(v)$ for each node
  $v$ to be $L$ plus its distance to the root. Setting $L \geq 2^h$
  gives the total $\ell_p$ error to be $O(L + 2^h)$, whereas any
  algorithm would incur cost at least $\approx L \cdot 2^{h}$.
\end{enumerate}
This suggests that the right extension to general edge-lengths
requires us to go beyond just parameterizing our results with the
maximum degree $\Delta$; this motivates our study of graphs with bounded doubling
dimension in
\S\ref{sec:full-information}.

\subsection{Gradient Information}
\label{sec:grad}

Consider the information model where the agent gets to see
\emph{gradient} information: each edge is imagined to be oriented
towards the endpoint with lower distance to the goal. The agent can
see some noisy version of these directions, and the error is the
number of edges with incorrect directions. We now show an example
where both the optimal distance and the error are $D$, but any
algorithm must incur cost $\Omega(2^D)$. Indeed, take a complete
binary tree of depth $D$, with the goal at one of the leaves. Suppose
the agent sees all edges being directed towards the root. The only
erroneous edges are the $D$ edges on the root-goal path. But any
algorithm must suffer cost $\Omega(2^D)$.

\end{document}